\documentclass[12pt]{article}

\usepackage[margin=1in]{geometry}
\usepackage{amsmath,amssymb,amsfonts,amsthm,mathtools,bm}
\usepackage{booktabs}
\usepackage{enumitem}
\usepackage{algorithm}
\usepackage{comment,color,xcolor}
\usepackage{algpseudocode}
\usepackage[round]{natbib}
\usepackage[colorlinks=true,citecolor=blue,linkcolor=blue,urlcolor=blue]{hyperref}

\newtheorem{theorem}{Theorem}

\newcommand{\RR}{\mathbb{R}}

\newcommand{\EE}{\mathbb{E}}
\newcommand{\PP}{\mathbb{P}}

\newcommand{\dd}{\mathrm{d}}
\newcommand{\eqd}{\overset{  \mathcal{L}  }{=}}
\newcommand{\ind}{\mathbf{1}}
\newcommand{\calA}{\mathcal{A}}
\newcommand{\calB}{\mathcal{B}}
\newcommand{\calC}{\mathcal{C}}
\newcommand{\calE}{\mathcal{E}}
\newcommand{\calF}{\mathcal{F}}
\newcommand{\calG}{\mathcal{G}}

\newcommand{\calL}{\mathcal{L}}
\newcommand{\calS}{\mathcal{S}}
\newcommand{\calT}{\mathcal{T}}

\newcommand{\sC}{\mathsf{C}}
\newcommand{\sE}{\mathsf{E}}
\newcommand{\sN}{\mathsf{N}}
\newcommand{\sU}{\mathsf{U}}
\newcommand{\sZ}{\mathsf{Z}}

\newcommand{\bbW}{\mathbb{W}}
\newcommand{\bbQ}{\mathbb{Q}}
\newcommand{\bbZ}{\mathbb{Z}}

\newcommand{\BridgeMax}{\textsc{BridgeMax}}
\newcommand{\RBMMax}{\textsc{RBMMax}}
\newcommand{\ICB}{\textsc{ICB}}

\title{Exact Simulation of Diffusions via Brownian Bridge Range Reconstruction}
\author{Radu Herbei \and Kumar Somnath}
\date{\today}

\begin{document}
\setcounter{algorithm}{-1}
\maketitle

\begin{abstract}
We develop an exact simulation algorithm for scalar diffusion paths and diffusion bridges when the Poisson potential is unbounded in both tails. The method reconstructs the realized range of a Brownian bridge proposal by sampling its maximum and location, together with the maxima and locations of the two adjacent restricted Brownian meanders. Conditional on this finite information, the remaining path decomposes into four conditionally independent interval-constrained Brownian bridges, which can be sampled exactly at the Poisson times required by the rejection test. In contrast to constructions based on an enclosing range layer, the
proposed representation retains the exact extrema and their locations. Our algorithm returns an exact finite-dimensional skeleton without time-discretization error and permits exact post-acceptance refinement at arbitrary finite collections of times. Numerical experiments validate the resulting finite-dimensional laws and identify the restricted-meander extremum simulation as the principal computational cost in the nonlinear example.
\end{abstract}



\section{Introduction}
\label{sec:intro}

Diffusion processes provide a fundamental class of continuous-time stochastic
models in applied probability, statistics, finance, biology, and the physical
sciences. Except in a small number of analytically tractable cases, however,
their transition densities and the distributions of path-dependent
functionals are not available in closed form. Simulation is therefore
typically based on a time-discretization scheme, such as the Euler--Maruyama
method or one of its higher-order variants. These methods converge as the
mesh is refined under appropriate regularity conditions, but every fixed
discretization introduces an approximation error that is distinct from the
Monte Carlo error and may be difficult to quantify in nonlinear models or for
functionals depending on the behavior of the path between grid points
\citep{KloedenPlaten2013}. Exact simulation addresses this issue by producing
a finite-dimensional random representation having precisely the law induced
by the stochastic differential equation (SDE), without introducing
time-discretization error. Such representations are useful both for direct
Monte Carlo calculation and as building blocks for likelihood-based
inference, diffusion-bridge simulation, and data-augmentation methods
\citep{BeskosRoberts2005,BeskosPapaspiliopoulosRobertsFearnhead2006}.

The principal general framework for exact simulation of scalar diffusions is
the Exact Algorithm (EA) methodology initiated by \citet{BeskosRoberts2005} and
developed further by
\citet{BeskosPapaspiliopoulosRoberts2006, BeskosPapaspiliopoulosRoberts2008}, all variants of a rejection sampling strategy. 
Consider a unit-diffusion process (after a Lamperti transformation, when necessary), defined using the following SDE
\[
    dX_t=\alpha(X_t)\,dt+dW_t,\quad 0\le t\le T.
\]
Girsanov's theorem and It\^o's formula then express the target diffusion law
relative to an endpoint-biased Brownian proposal law in terms of the
acceptance functional
\[
    a(B)
    =
    \exp\left\{
        -\int_0^T \phi(B_t)\,dt
    \right\},
\]
where $\phi$ is a non-negative function determined by the drift, and $B$ is the proposed path, both carefully defined later. The ideal rejection sampling algorithm requires simulating the Brownian path $B$ (entirely) and a coin  $\sC\sim {\rm Bernoulli}(a(B))$. Both steps are technically impossible: $B$ is a continuous path and exact evaluation of $a(B)$ is not possible. These issues are circumvented by (1) only revealing finite-dimensional information about $B$ \emph{as necessary}, and (2) observing that simulating $\sC$ is equivalent to simulating a carefully defined unit-rate Poisson process, and counting how many points fall below the graph of $t\mapsto \phi(B_t)$, acceptance being equivalent to ``no Poisson points fall below the graph''. This is a computationally tractable problem and can be done exactly.

The successive versions of the Exact Algorithm differ in how $B$ is revealed and how the domain of the unit-rate Poisson process is defined. The algorithm of \cite{BeskosRoberts2005}, henceforth called EA1 assumes that $\phi$ is bounded globally, so a deterministic Poisson rectangle and ordinary Brownian-bridge interpolation are sufficient \citep{BeskosRoberts2005}. The subsequent algorithmic improvements described in \cite{BeskosPapaspiliopoulosRoberts2006,BeskosPapaspiliopoulosRoberts2008}, henceforth named EA2 and EA3 respectively, relax the boundedness of $\phi$ assumption, and replace it with a one-sided tail control paired with a path extrema simulation (EA2) or a bounding layer simulation and no tail control (EA3). A brief description of these strategies is given later.

Related developments have pursued several complementary ways of obtaining
finite information about an otherwise infinite-dimensional path. The
$\varepsilon$-strong construction of
\citet{BeskosPeluchettiRoberts2012} generates successively tighter upper and
lower path envelopes. \citet{PollockEtAl2016} developed unbounded and adaptive
exact algorithms based on refined layer information and placed exact path
restoration within a general finite-dimensional framework for diffusions and
jump diffusions. \citet{ChenHuang2013} instead used sequential spatial
localization, decomposing a Brownian-driven diffusion into pieces confined
to bounded neighborhoods. Although these constructions differ
substantially, each must solve the same underlying problem: identify finite
random information that both bounds the Poisson potential along the entire
proposal and leaves a conditional path law from which the values required by
the rejection test can be sampled exactly.

This paper gives a direct solution to this problem for the fixed-horizon
scalar EA3 setting. Rather than sampling an interval known only to contain
the Brownian proposal, we reconstruct its realized minimum and maximum
exactly. The construction may be viewed as extending the extremum
decomposition underlying EA2 by one additional step. We use a sequential approach: for a proposed Brownian bridge $B$, we first simulate its maximum $\overline B$ and its location $\tau^+$. The path fragments on the two sides of $\tau^+$,
after the appropriate reflection and time reversal, are conditionally
independent restricted Brownian meanders \citep{Williams1974,PitmanYor1996}. We then simulate the maximum and its location for each of these two depth processes. These secondary extrema can be used then to recover the global minimum \(\underline B\) and its location $\tau^-$. The complete realized range $[\underline B,\overline B]$ is therefore known without simulating the full Brownian path.

The construction becomes computationally explicit through two recently
developed exact Brownian sampling primitives. The joint maximum and
maximum-location sampler for restricted Brownian meanders is provided by
\citet{SomnathHerbei2026MaxLocation}. Conditional on the endpoint, the extrema, and
their locations, the unrevealed path $B$ separates into four conditionally
independent Brownian bridges constrained to remain within explicitly known
intervals. Their values at arbitrary finite collections of times can be
generated using the exact interval-constrained interpolation method of
\citet{HerbeiSomnath2026Interval}. Consequently, the exact range supplies a finite
bound for the proposal path, while the four-piece factorization
supplies exactly the conditional values needed at the Poisson times. This
yields a genuine rejection sampler from the diffusion law.

A closely related decomposition appears in \citet{ChenHuang2012}, who
decompose a Brownian proposal at its maximum and use Brownian meanders to
construct unbiased importance-sampling estimators for functionals involving
diffusion extrema. When both extrema are required, the opposite extremum is
recovered through the maxima of the two meander segments. Their construction
produces an importance-sampling weight and is directed toward unbiased
estimation.
The present work additionally simulates the locations of the secondary
extrema. Those locations divide the proposal into the four
interval-constrained bridge pieces required for exact Poisson thinning and
for subsequent conditional refinement. Thus, although the underlying
maximum decomposition is related, the probabilistic output and the role of
the decomposition are different.

The contribution of this paper is therefore an exact simulation algorithm for a large class of diffusions, based on exact reconstruction of the range of each delivered path. First, we provide a constructive path factorization that reconstructs the exact range of a Brownian bridge from
one Brownian-bridge extremum and two restricted-meander extrema. Second, we
combine this factorization with interval-constrained interpolation to obtain
an exact algorithm for both unconditioned diffusions and diffusion bridges,
and prove that its accepted output permits exact post-acceptance refinement.
Third, we give complete implementations for an Ornstein--Uhlenbeck diffusion
and a nonlinear double-well diffusion. In both examples the Poisson
potential is unbounded in both directions, so neither EA1 nor either
one-sided version of EA2 applies. The Ornstein--Uhlenbeck model provides
closed-form transition and acceptance probabilities against which the
implementation can be checked. The double-well model illustrates the method
when no elementary transition density is available and also exhibits the
discretization distortion that may arise from coarse Euler--Maruyama
simulation. The timing experiments document the computational allocation of
the method and identify the range-reconstruction and constrained 
interpolation steps that govern its practical cost.

The remainder of the paper is organized as follows. Section~\ref{sec:prelim} introduces the
path-space formulation, reviews the Wiener--Poisson factorization and the
EA1--EA3 constructions, and describes the Brownian sampling primitives used
throughout. Section~\ref{sec:mainsampler} develops the exact range reconstruction, identifies
the four conditional bridge components, states the complete algorithm, and
proves its correctness and refinement properties. Section~\ref{sec:sims} gives the
Ornstein--Uhlenbeck and double-well implementations and the accompanying
numerical studies. Section~\ref{sec:disc} concludes with a discussion of the scope,
computational characteristics, and possible extensions of the method. The
Appendix verifies the standing assumptions for the double-well example.

\section{Preliminaries}
\label{sec:prelim}

Fix a finite horizon \(T\in(0,\infty)\) and an initial value \(x\in\RR\). Let
\(
        (\Omega,\calA,\{\calA_t\}_{0\le t\le T}, \PP)
\)
be a complete filtered probability space. The time $T$ will be fixed from now on and we suppress the dependence on $T$ to avoid cumbersome notation.
A point of this space is denoted by \(\xi\in\Omega\).  We assume that \((\Omega,\calA,\PP)\) is rich enough to support all primitive random objects required by the algorithms developed below, including countably many independent copies of them.  We write \(\eqd\) to denote equality in distribution. We use \(\sU_1,\sU_2,\ldots\) to denote iid Uniform\((0,1)\) random variables, \(\sE_1,\sE_2,\ldots\) for iid Exponential\((1)\) random variables, and  \(\sN_1,\sN_2,\ldots\) for iid standard normal random variables. These sequences are independent of each other.

On the same space let \(W=\{W_t:0\leq t\leq T\}\) be an $\{\calA_t\}$-standard Brownian motion, independent of the auxiliary variables just described. For \(x\in\RR\), write
\(
        W_t^x=x+W_t,
\)
so that \(W^x\) is Brownian motion started at \(x\). We consider an adapted continuous process $X$ defined via the following one-dimensional unit diffusion model: 
\begin{equation}
\label{eq:sde}
        \dd X_t=\alpha(X_t)\,\dd t+\dd W_t,
        \qquad 0\leq t\leq T,
        \qquad X_0=x\in\RR,
\end{equation}
where \(\alpha:\RR\to\RR\) is the drift coefficient. We assume that \eqref{eq:sde} has a global, weakly unique solution with continuous sample paths.  Standard references for Brownian motion, weak solutions, and diffusion laws include \citet{KaratzasShreve1998} and \citet{Oksendal2013}.
For a scalar diffusion with non-unit diffusion coefficient,
\[
        \dd V_t=b(V_t)\,\dd t+\sigma(V_t)\,\dd W_t,
\]
with \(\sigma>0\) sufficiently smooth, the Lamperti transform \citep{Lamperti1964}
\[
        X_t=\eta(V_t),\qquad \eta(v)=\int^v\frac{1}{\sigma(r)}\,\dd r,
\]
reduces the model to \eqref{eq:sde}, with
\begin{equation*}
        \alpha(u)=\frac{b(\eta^{-1}(u))}{\sigma(\eta^{-1}(u))}
        -\frac12\sigma'(\eta^{-1}(u)).
\end{equation*}
We consider scalar diffusions for which the Lamperti transformation is a $C^2$ bijection from the original space onto $\RR$. Under these conditions, an exact skeleton for \(X\) is transformed back by \(V_t=\eta^{-1}(X_t)\).  The remainder of the paper is therefore written for \eqref{eq:sde} without loss of generality for scalar diffusions satisfying the Lamperti conditions.

Our stated goal is to simulate an \emph{exact} discrete trajectory (or path) of the diffusion \eqref{eq:sde}. This is known to be a challenging task, and we will clarify this statement below, once the notation and setup are complete. The route to this goal is rejection sampling on the path space, which we describe next. Let
\(
        \calC=C([0,T],\RR)
\)
be the space of continuous real-valued functions on \([0,T]\), endowed with the supremum norm \(\|\omega\|_\infty=\sup_{0\leq t\leq T}|\omega(t)|\). An element \(\omega\in\calC\) is a deterministic continuous path. We write either \(\omega(t)\) or \(\omega_t\) for its value at time \(t\). The canonical coordinate maps are
\[
        B_t:\calC\to\RR,
        \qquad B_t(\omega)=\omega(t),
        \qquad 0\leq t\leq T,
\]
and the canonical \(\sigma\)-field and filtration are
\[
        \calF=\calB(\calC) = \sigma\{B_t:0\leq t\leq T\},
        \qquad
        \calF_t=\sigma\{B_s:0\leq s\leq t\},\quad 0\leq t\leq T.
\]
If \(\nu\) is a probability measure on \((\calC,\calF)\), the coordinate collection \(B=\{B_t:0\leq t\leq T\}\) is a stochastic process on \((\calC,\calF,\nu)\). All path functionals used below, including endpoints, extrema, and hitting times, are measurable functions on this canonical space. For a canonical process $B$ and a compact interval $J\subseteq[0,T]$ define the minimum, maximum and their locations as
\begin{align}
\underline{B}_J
    &:= \inf_{u\in J} B_u,
&
\overline{B}_J
    &:= \sup_{u\in J} B_u,
    \label{eq:minmax1}
\\
\tau^{-}_J
    &:= \inf\bigl\{u \in J : B_u = \underline{B}_J\bigr\},
&
\tau^{+}_J
    &:= \inf\bigl\{u \in J : B_u = \overline{B}_J\bigr\}.
    \label{eq:minmax2}
\end{align}
When $J = [0,T]$ we suppress the interval and simply write $\underline B$, $\overline B$, $\tau^-$, and $\tau^+$. 

If \(Y=\{Y_t:0\leq t\leq T\}\) is any continuous process on $(\Omega, \calA, \PP)$, we also write \(Y\) for the measurable map into \(\calC\), thus $Y(\xi):[0,T]\to \RR$ is defined as \(Y(\xi)(t)=Y_t(\xi)\). Its law is the push-forward measure \(\calL(Y)=\PP \circ Y^{-1}\). In particular, the Wiener measure started at \(x\) is
\(
        \bbW^x=\calL(W^x) = \PP\circ (W^x)^{-1}.
\)
Under \(\bbW^x\), the canonical process satisfies \(B_0=x\), and \(B_t-B_s\sim N(0,t-s)\), independently of \(\sigma\{B_r:0\leq r\leq s\}\), for \(0\leq s<t\leq T\). For \(z\in\RR\), \(\bbW^{x\to z}\) denotes the Brownian bridge measure, as induced by the process 
\begin{equation*}
x +\frac{t}{T}(z-x) + W_t - \frac{t}{T}W_T\ ,\quad 0\le t\le T\ .
\end{equation*}
The diffusion law induced by \eqref{eq:sde} is
\(
        \bbQ^x=\calL(X) = \PP \circ X^{-1},
\)
so that under \(\bbQ^x\), \(B\) is the canonical version of the diffusion. 

We aim to generate a finite-dimensional object $(\calS, \Lambda)$ defined such that:
\begin{enumerate}
\item $\cal S$ is a  \emph{time-value skeleton}
\begin{equation}
        \mathcal S=(N, \{ (t_i,X_i):1\leq i\leq N\})
        \qquad 0\leq t_1<\cdots<t_N\leq T,
        \label{eq:skeleton}
\end{equation}
where \(N, \{(t_i,X_i), 1\le i\le N\}\) may be random;
\item $\Lambda$ is a finite-dimensional \emph{auxiliary component}, consisting of a finite collection of path functionals under the target law.
\item There exists a path $B\sim \bbQ^x$ such that $X_i = B_{t_i}$ (i.e., $\calS$ is a skeleton of $B$), and the conditional law of B given $(\calS, \Lambda)$ permits exact refinement of $\calS$ at any additional finite collection of deterministic times.
\end{enumerate}

 The tool used to accomplish this is a variant of the well-known rejection sampling algorithm, whose output will be a diffusion path revealed only through a skeleton of the form \eqref{eq:skeleton}.
Once a path is accepted and its skeleton is revealed, it must carry enough conditional information to allow further refining of the accepted path at any additional finite set of times, from the correct regular conditional law. Throughout, exact means \emph{distributionally exact} under ideal random-number generation and arithmetic: no time-discretization error is introduced. Practical implementations remain subject to pseudo-random-number generation and finite-precision arithmetic. This is an important issue, as it is well-known that the finite-dimensional distributions induced by \eqref{eq:sde} are generally unavailable in closed form. Conventional simulation of these finite-dimensional distributions is usually based on time-discretization approximations \citep{KloedenPlaten2013}.

Rejection sampling on the path space can be described as follows. If $\mu, \nu$ are probability measures on $(\calC, \calF)$, and $\mu\ll\nu$, and if $a:\calC\to[0,1]$ is proportional to \(\dd\mu/\dd\nu\), then a proposal \(Y\sim\nu\) accepted with conditional probability \(a(Y)\) has law \(\mu\). As described, this algorithm is essentially impossible to implement since $Y$ is of infinite dimension. Nevertheless, under certain assumptions, exact sampling of a discrete skeleton \eqref{eq:skeleton} is possible. As mentioned in Section~\ref{sec:intro}, we follow the EA literature \citep{BeskosRoberts2005,BeskosPapaspiliopoulosRoberts2006,BeskosPapaspiliopoulosRoberts2008}. A brief description of the setup and conditions of these algorithms follows next.

Throughout the paper we impose the following standing assumptions, which are commonly used in the EA papers:
\begin{enumerate}[label=(A\arabic*),ref=A\arabic*,leftmargin=2.2em,start=0]
\item \label{assump:rn} \(\bbQ^x\ll\bbW^x\) with Radon-Nikodym derivative
\[
\frac{\dd \bbQ^x}{\dd \bbW^x}= 
\exp\left\{
\int_0^T\alpha(B_t)\,\dd B_t
        -\frac12\int_0^T\alpha^2(B_t)\,\dd t
\right\}\ .
\]
\item \label{assump:diff} \(\alpha\) is continuously differentiable.
\item \label{assump:endpoint} With \(A(u)=\int_0^u\alpha(v)\,\dd v\),
\[
        c(x,T)=\int_\RR\exp\left\{A(u)-\frac{(u-x)^2}{2T}\right\}\,\dd u<\infty .
\]
\item \label{assump:lower} The function \(\alpha^2+\alpha'\) is bounded below.
\end{enumerate}
A sufficient condition for Assumption \eqref{assump:rn} is Novikov's condition
\begin{equation*}
\displaystyle \EE_{\bbW^x}\left[\exp\left\{\frac12\int_0^T\alpha^2(B_t)\,\dd t\right\}\right]<\infty\ .
\end{equation*}
Assumption \eqref{assump:diff} allows the stochastic integral in the Girsanov density to be rewritten by Ito's formula. Assumption \eqref{assump:endpoint} makes the endpoint-biased Brownian proposal normalizable, and Assumption \eqref{assump:lower} permits a non-negative potential in a Poisson rejection step, as we explain next.

Based on \eqref{assump:lower}, choose \(k_L\in\RR\) such that
\[
        k_L\leq \frac12\{\alpha^2(u)+\alpha'(u)\},\qquad u\in\RR,
\]
and define
\begin{equation}
\label{eq:phi}
        \phi(u)=\frac12\{\alpha^2(u)+\alpha'(u)\}-k_L\ge 0,
        \qquad u\in\RR
\end{equation}
 Using assumption \eqref{assump:endpoint}, define the endpoint density
\begin{equation}
\label{eq:h_endpoint}
        h(u)=\frac{1}{c(x,T)}\exp\left\{A(u)-\frac{(u-x)^2}{2T}\right\},
        \qquad u\in\RR
\end{equation}
and the biased Brownian proposal law as the probability measure \(\bbZ^x\) on \((\calC,\calF)\) given by
\begin{equation}
\label{eq:Z_def}
        \bbZ^x(F)=\int_\RR\bbW^{x\to u}(F)h(u)\,\dd u,
        \qquad F\in\calF.
\end{equation}
Equivalently, to sample \(B\sim\bbZ^x\), one first samples \(Z\sim h\) and then, conditionally on \(Z=z\), samples a Brownian bridge from \(x\) to \(z\) over \([0,T]\).

\subsection{Diffusion law and endpoint-biased Brownian proposal}
\label{sec:endpoint-proposal}

The stochastic integral in Assumption \eqref{assump:rn} is interpreted under \(\bbW^x\). By It\^o's formula and Assumption \eqref{assump:diff},
\[
        \int_0^T\alpha(B_t)\,\dd B_t
        =A(B_T)-A(x)-\frac12\int_0^T\alpha'(B_t)\,\dd t,
\]
and hence
\begin{equation}
\label{eq:rnqw}
        \frac{\dd \bbQ^x}{\dd \bbW^x}
        =\exp\left\{A(B_T)-A(x)-\frac12\int_0^T\{\alpha^2(B_t)+\alpha'(B_t)\}\,\dd t\right\}.
\end{equation}
The role of \(\bbZ^x\) is to remove the endpoint factor \(\exp\{A(B_T)\}\). From \eqref{eq:Z_def} and the Gaussian endpoint density of Brownian motion under \(\bbW^x\),
\begin{equation}
\label{eq:rnzw}
        \frac{\dd \bbZ^x}{\dd \bbW^x}
        =\frac{h(B_T)}{(2\pi T)^{-1/2}\exp\{-(B_T-x)^2/(2T)\}}
        =\frac{\sqrt{2\pi T}}{c(x,T)}\exp\{A(B_T)\}.
\end{equation}
Combining \eqref{eq:rnqw} and \eqref{eq:rnzw}, and using the definition of \(\phi\) in \eqref{eq:phi}, all dependence on \(B_T\) cancels. Consequently,
\begin{equation}
\label{eq:rnqz}
        \frac{\dd \bbQ^x}{\dd \bbZ^x}
        =\frac{\exp\left\{-\int_0^T \phi(B_t)\,\dd t\right\}}
        {\EE_{\bbZ^x}\left[\exp\left\{-\int_0^T \phi(B_t)\,\dd t\right\}\right]}.
\end{equation}
The denominator in \eqref{eq:rnqz} is an unknown normalizing constant and is irrelevant for rejection sampling. 
Define
\begin{equation}
\label{eq:acceptprob}
a(\omega)
=
\exp\left\{-\int_0^T \phi(B_t(\omega))\,dt\right\}.
\end{equation}
Then \eqref{eq:rnqz} can be written as
\[
\frac{d\mathbb Q^x}{d\mathbb Z^x}(\omega)
=
\frac{a(\omega)}
{\mathbb E_{\mathbb Z^x}[a(B)]}.
\] 
Equation \eqref{eq:rnqz} is the fundamental component of the rejection sampling algorithm used by the Exact Algorithms and subsequently in this paper. Since $\phi\ge 0$, in its simplest version, the structure we will use is:
\begin{algorithm}[!ht]
\caption{Rejection sampling on the path space}
\label{alg:path_rejection}
\begin{algorithmic}[1]
\State Draw $B\sim \bbZ^x$;
\State Accept with probability
\begin{equation*}
a(B) = \exp\left\{
-\int_0^T\phi(B_t)\ \dd t 
\right\}\in [0,1]\ .
\end{equation*}
\end{algorithmic}
\end{algorithm}

By \eqref{eq:rnqz}, Algorithm~\ref{alg:path_rejection} produces a path $B\sim \bbQ^x$. However, as discussed, neither of the two steps above can be done exactly in practice. Instead, we will replace Step 1 with the exact simulation of a finite-dimensional object, whose law is consistent with $\bbZ^x$, and instead of evaluating $a(B)$, we simulate a finite-dimensional event having the exact same probability of success. The necessary pieces for this procedure will be added later.

For $z\in \RR$ the diffusion bridge law $\bbQ^{x\to z}$ is defined as 
\begin{equation*}
    \bbQ^{x\to z}(d\omega) = 
    \frac{a(\omega)}{\EE_{\bbW^{x\to z}}[a(B)]}
     \bbW^{x\to z}(d\omega)
\end{equation*}
and it then satisfies
\[
\frac{d\mathbb Q^{x\to z}}{d\mathbb W^{x\to z}}(\omega)
=
\frac{a(\omega)}
{\mathbb E_{\mathbb W^{x\to z}}[a(B)]}.
\]

The law $\bbQ^{x\to z}$ agrees with a regular conditional version of $\bbQ^x(\cdot\mid B_T=z)$.
Thus the unconditional and bridge algorithms use the same pathwise
acceptance probability; they differ only in whether the endpoint is
sampled from $h$ or fixed at $z$.

\subsection{Poisson thinning and the finite-dimensional acceptance event}
We first explain how Step 2 in Algorithm~\ref{alg:path_rejection} is accomplished. The EA algorithms convert  this accept-reject step into a computable test by Poisson thinning. For now we imagine that the proposed path $B$ is available. Suppose that, conditionally on  \(B=\omega\), a finite bound
\begin{equation*}
        M_\phi(\omega)\geq \sup_{0\leq t\leq T}\phi(\omega_t)
\end{equation*}
can be computed. Simulate a homogeneous Poisson process \(\Phi\) of unit intensity on the rectangle \([0,T]\times[0,M_\phi(\omega)]\). Write
\[
        \Phi=\{(S_i,V_i):1\leq i\leq K\}.
\]
The event that no Poisson point falls below the graph \(t\mapsto \phi(\omega_t)\) is
\begin{equation*}
        \calE(\omega,\Phi)=\bigcap_{i=1}^K\{V_i\ge\phi(\omega_{S_i})\}.
\end{equation*}
In practice, if $M_\phi(\omega)>0$, generate
\[
K\mid B=\omega
\sim
\operatorname{Poisson}\{T M_\phi(\omega)\},
\]
and, conditional on $K$, generate $\{(S_i, V_i)\ :\ 1\le i\le K\}$ independent points uniformly over
$[0,T]\times[0,M_\phi(\omega)]$. Accept the path $B$ if none of these points lies
below the graph of $t\mapsto\phi(B_t(\omega))$. Indeed, evaluating
\[
\Sigma=\sum_{i=1}^K \ind\{V_i<\phi(B_{S_i})\}
\qquad\Rightarrow\qquad
\Sigma\mid B = \omega \sim {\rm Poisson}\left(\int_0^T \phi(B_t(\omega)) \dd t\right)
\]
and thus
\[
\Pr(\operatorname{accept}\mid B=\omega) = \Pr(\Sigma= 0\mid B = \omega) =
\exp\left\{-\int_0^T\phi(B_t(\omega) )\dd t\right \}.
\]

If $M_\phi(\omega)=0$, then the nonnegativity of $\phi$ and the defining
upper-bound property imply that
\[
\phi(B_t(\omega))=0,
\qquad 0\leq t\leq T.
\]
Consequently, $ \int_0^T \phi(B_t(\omega)) \dd t =0$ and
\(
K\mid B = \omega\sim\operatorname{Poisson}(0),
\)
so that $K=0$ almost surely and the proposal is accepted automatically.
Thus the identity
\[
\Pr(\operatorname{accept}\mid B = \omega)
=
\exp\left\{
-\int_0^T\phi(B_t(\omega))\,dt
\right\}
\]
also holds when $M_\phi(\omega)=0$.

This construction illustrates that the accept/reject event can be simulated exactly even without full knowledge of the entire proposed path $\omega$. Instead, it is sufficient to know the following elements (in order):
\begin{enumerate}
\item the bound $M_\phi(\omega)$;
\item the Poisson count $K\sim {\rm Poisson}\{TM_\phi(\omega)\}$;
\item the values $\phi(\omega_{S_i})$ for iid ${\rm Uniform}(0,T)$ variates $S_i, i=1, \ldots, K$;
\end{enumerate}

As such, instead of knowing the entire path $\omega$, it is sufficient to ``reveal'' only certain aspects of it, so that the elements above can be computed exactly. This observation also addresses the issue with Step 1 of Algorithm~\ref{alg:path_rejection}: it is not necessary to sample the entire $B$, but just enough of it to permit evaluation of $M_\phi(\omega)$ and $\phi(\omega_{S_i})$.
The Exact Algorithms are all instances of this same rejection principle and strategy. Since simulating $K$ is trivial, the algorithms differ only in how much information about the proposed Brownian path is revealed in order to construct the height \(M_\phi(\omega)\) of the Poisson rectangle and to interpolate the path at the Poisson times.

\textbf{EA1}: In the simplest case, a global boundedness condition is assumed:
\begin{equation*}
\tag{A4}
\label{eq:A4}
        \exists M<\infty\quad\text{such that}\quad \phi(u)\leq M,\qquad u\in\RR
\end{equation*}
Now the Poisson rectangle is deterministic: \([0,T]\times[0,M]\). The EA1 algorithm is consequently very simple. Draw \(Z\sim h\), construct the Brownian bridge proposal from \(x\) to \(Z\), simulate \(\Phi\) on \([0,T]\times[0,M]\), sample the bridge at the Poisson times using ordinary Brownian bridge interpolation, and accept if no point lies below \(\phi\). Conditional on the proposed path, the acceptance probability is exactly \eqref{eq:acceptprob}; hence the accepted proposal has law \(\bbQ^x\) by \eqref{eq:rnqz}. The bridge version replaces \(Z\sim h\) by a fixed endpoint \(Z=z\) and uses  the same acceptance probability \eqref{eq:acceptprob}. EA1 is the cleanest version of the method, but condition \eqref{eq:A4} is restrictive. If \(\phi(u)\to\infty\) in either tail, a deterministic rectangle cannot be used. Many statistically important diffusions fail \eqref{eq:A4}; for instance, for the Ornstein--Uhlenbeck drift \(\alpha(u)=\theta_1-\theta_2u\), one obtains \(\phi(u)\) proportional to a quadratic function after shifting by \(k_L\).

\textbf{EA2:} Relaxing \eqref{eq:A4}, now assume that $\phi$ is bounded in one tail, say
\begin{equation*}
\tag{A5}
\label{eq:A5}
        \limsup_{u\to\infty}\phi(u)<\infty.
\end{equation*}
The symmetric version assumes \(\limsup_{u\to -\infty}\phi(u)<\infty\). Suppose \eqref{eq:A5} holds and for the Brownian bridge proposal $B\sim \bbW^{x\to Z}$ consider the minimum and its location $(\underline B, \tau^-)$ as in \eqref{eq:minmax1} and \eqref{eq:minmax2}.

If \((\underline B,\tau^-)\) are available, the proposed path is known to lie in \([\underline B,\infty)\) a.s. Under \eqref{eq:A5} and local boundedness of \(\phi\),
\begin{equation*}
        M_\phi=\sup_{u\geq \underline B}\phi(u)<\infty\quad a.s. 
\end{equation*}
is a valid path-dependent Poisson height. 

Conditional on \((\underline B,
\tau^-)\) and the endpoint $Z = z$,  the two processes
\begin{align*}
        R_1(s)&=B_{\tau^--s}-\underline B,\qquad 0\leq s\leq \tau^-,\quad R_1(0) = 0, R_1(\tau^-) = x-\underline{B}\\
        R_2(s)&=B_{\tau^-+s}-\underline B,
        \qquad 0\leq s\leq T-\tau^-,\quad R_2(0) = 0, R_2(T-\tau^-) = z-\underline{B}\ ,
\end{align*}
are independent three-dimensional Bessel bridges, or equivalently restricted Brownian meanders, with terminal values \(x-\underline B\) and \(z-\underline B\). The Brownian proposal can therefore be filled at the Poisson times by Bessel-bridge interpolation. This is the retrospective part of EA2: only those values needed for the Poisson test are generated.

If instead \(\limsup_{u\to-\infty}\phi(u)<\infty\), the construction is reflected. One simulates the maximum \(\overline B\) and its location, obtains a finite bound \(\sup_{u\leq \overline B}\phi(u)\), and fills the decomposed pieces below the maximum by reflected Bessel-bridge interpolation. EA2 is therefore an extremum-based exact algorithm: one extremum is enough because one tail of \(\phi\) is already controlled.

\textbf{EA3:}  Now remove the global boundedness condition required by EA1 and the
one-sided boundedness condition required by EA2. Conditional on the
proposal endpoint \(Z=z\), let
\(
        B\sim\mathbb W^{x\to z}
\)
be a Brownian bridge proposal from \(x\) to \(z\) over \([0,T]\), and consider the two extrema $(\underline B, \overline B)$ as in \eqref{eq:minmax1}, \eqref{eq:minmax2}. If \(\phi\) is unbounded in both tails, knowledge of only \(\underline B\) or only \(\overline B\) does not provide a finite upper bound for \(t\mapsto\phi(B_t)\). The original EA3 construction of \citet{BeskosPapaspiliopoulosRoberts2008} instead reveals a random compact interval containing the entire Brownian bridge.

Let
\(
        x_\wedge=x\wedge z,\ 
        x_\vee=x\vee z,
\)
and choose a strictly increasing sequence
\(
        0=a_0<a_1<a_2<\cdots
\)
such that \(a_i\to\infty\). Define the lower and
upper boundaries
\(
        \ell_i=x_\wedge-a_i,
        \ 
        u_i=x_\vee+a_i,\ i\geq 0
\).
For \(i\geq1\), introduce the events
\begin{align*}
\mathcal U_i
 &=
 \left\{
        u_{i-1}\leq\overline B<u_i
 \right\}
 \cap
 \left\{
        \underline B>\ell_i
 \right\},
 \\
\mathcal L_i
 &=
 \left\{
        \ell_i<\underline B\leq\ell_{i-1}
 \right\}
 \cap
 \left\{
        \overline B<u_i
 \right\},
\end{align*}
and set
$\mathcal D_i=\mathcal U_i\cup\mathcal L_i$.
Up to events of Brownian bridge probability zero, the collection
\(\{\mathcal D_i:i\geq1\}\) is a partition of the path space. Define
the discrete random variable \(I=I(B)\) by
\(
        \{I=i\}=\mathcal D_i.
\)
The variable \(I\) is called the \emph{layer} of the Brownian bridge. The layer has a particularly simple interpretation. For every
\(i\geq1\),
\begin{equation*}
\begin{split}
        \{I\leq i\}
        &=
        \left\{
        \ell_i<B_t<u_i
        \text{ for every }0\leq t\leq T
        \right\}
        =
        \left\{
        \underline B>\ell_i,\;
        \overline B<u_i
        \right\}.
\end{split}
\end{equation*}
Thus \(I\) is the index of the first interval
\((\ell_i,u_i)\) in the nested sequence
\(
        (\ell_1,u_1)\subset(\ell_2,u_2)\subset\cdots
\)
that contains the complete bridge path. On the event \(\{I=i\}\),
we have \( \ell_i<B_t<u_i, 0\leq t\leq T\),
and therefore
 $M_{\phi,i}
        =
        \sup_{\ell_i\leq u\leq u_i}\phi(u)
        <\infty
$
is a valid height for the Poisson rejection rectangle. More generally,
any valid upper bound for $M_{\phi,i}$ may be used.
The random variable $I$ can be simulated exactly, again, using Devroye's series approach. Its cdf can be represented as an infinite series which can be approximated from above and below by monotone sequences.

After \(I=i\) has been sampled, EA3 must simulate the Brownian bridge
at the Poisson times from
\(
        \mathbb W_i^{x\to z}
        :=
        \mathbb W^{x\to z}(\,\cdot\mid\mathcal D_i).
\)
This conditional law is referred to as a \emph{layered Brownian
bridge}. The obvious rejection sampler: proposing an ordinary Brownian bridge and retaining it
when its layer is \(i\) is not satisfactory. Conditional on \(I=i\),
the expected number of such proposals is
\(1/\mathbb W^{x\to z}(\mathcal D_i)\), which has infinite unconditional expectation. \cite{BeskosPapaspiliopoulosRoberts2008} develop a more efficient sampler, and we omit all the details here.

At a high level, the EA3 procedure is therefore as follows. First draw the
proposal endpoint \(Z\sim h\); for exact diffusion-bridge simulation,
set \(Z=z\) equal to the prescribed endpoint. Conditional on \(Z\),
simulate the layer \(I\). Compute the finite layer-dependent
bound
\(
        M_{\phi,I}
        =
        \sup_{\ell_I\leq u\leq u_I}\phi(u),
\)
and simulate a unit-rate Poisson process on
\(
        [0,T]\times[0,M_{\phi,I}].
\)
Next, use the layered Brownian bridge rejection sampler to generate the bridge values \(B_{S_1},\ldots,B_{S_K}\) at the Poisson
times, conditional on \(Z\) and \(I\). Finally, accept the proposal if no Poisson points fall below the graph of $\phi$, and otherwise restart the construction.

\subsection{Brownian building blocks}
Before detailing the main contributions of this work we review two fundamental building blocks necessary for the algorithms we develop below. 

\subsubsection{Interval-constrained Brownian interpolation}

For a non-degenerate interval \([\ell,u]\) and two points $x,z\in [\ell, u]$, set \(a=u-\ell\). 
Let \(D=(0,a)\), and $r\in D$. Set \(p^D(y;r,t)\) to be the Brownian transition density killed on exiting \(D\) (started at $r$). It is known \citep{BorodinSalminen2012, revuz2013continuous} that
\begin{align*}
        p^D(y;r,t)
        &=\sum_{n=-\infty}^{\infty}\frac{1}{\sqrt{2\pi t}}
        \left[\exp\left\{-\frac{(y+2na-r)^2}{2t}\right\}
        -\exp\left\{-\frac{(y+2na+r)^2}{2t}\right\}\right]  \\
        &=\frac{2}{a}\sum_{n=1}^{\infty}
        \sin\left(\frac{\pi n r}{a}\right)
        \sin\left(\frac{\pi n y}{a}\right)
        \exp\left\{-\frac{n^2\pi^2t}{2a^2}\right\},\quad 0<y<a\ . \nonumber
\end{align*}
The first expression is the reflection representation; the second is the Dirichlet eigenfunction representation. These complementary series are the two analytic forms exploited by the interval-constrained sampler. For any interior time point \(0<t<T\), interior end points $x,z\in(\ell, u)$, and $\ell<y<u$, the conditional density of an interval-constrained Brownian bridge from $x$ to $z$ is
\begin{equation}
\label{eq:icb_density}
    \bbW^{x\to z}\{B_t\in \dd y\mid \ell\leq B_s\leq u,
        0\leq s\leq T\}
        \propto p^D(y-\ell;x-\ell, t)p^D(z-\ell;y-\ell, T-t)\,\dd y.
\end{equation}
Equivalently, the unnormalized interpolation density is a product of two killed transition densities. This is the Markov factorization that underlies the exact interpolation algorithms of \citet{HerbeiSomnath2026Interval}. Boundary endpoints $x\downarrow \ell$,$x\uparrow u$ and $z\downarrow \ell$,$z\uparrow u$ are handled by limits from the interior and are treated separately. For example, when the left endpoint is \(x=\ell\), the limiting density factor is proportional to \citep[see][]{durrett1977functionals}
\begin{equation}
\label{eq:h_boundary}
        \sum_{n=-\infty}^{\infty}(2na+(y-\ell))
        \exp\left\{-\frac{(2na+(y-\ell))^2}{2t}\right\},
        \qquad 0\leq y-\ell\leq a.
\end{equation}

For $\Delta>0$, $\ell<u$, and $x,z\in[\ell,u]$, let
\(
\mathbb W_{\Delta;[\ell,u]}^{x\to z}
\)
denote the law of a Brownian bridge from $x$ to $z$ over $[0,\Delta]$,
conditioned to remain in $[\ell,u]$. When $x,z\in(\ell,u)$, this is the
usual conditional bridge law determined by the killed transition
density above. 
When one or both endpoints lie in $\{\ell,u\}$, it is
the corresponding weak entrance or exit limit of the interior laws.
Under each of these boundary-limit laws,
\(
\ell<B_s<u,
\)
for all $0<s<\Delta,$ almost surely. \cite{HerbeiSomnath2026Interval} give an exact sampler for the finite-dimensional
distributions of $\mathbb W_{\Delta;[\ell,u]}^{x\to z}$ in all interior
and boundary cases. We write
\[
\operatorname{ICB}(\Delta,x,z,\ell,u;\mathcal T)
\]
for this primitive, where $\mathcal T\subset(0,\Delta)$ is finite.
The primitive returns the bridge values at the times in $\mathcal T$.

\subsubsection{Extrema of Brownian bridges and restricted Brownian meanders}

The second necessary ingredient is exact simulation of extrema and their locations for Brownian bridges and restricted Brownian meanders. Under $\bbW^{0\to r}$, the pair $(\overline B, \tau^+)$ has density
\begin{equation}
\label{eq:bridge_max_density}
        p_r(m,t)=\frac{\sqrt{2\pi T}\,\exp\{r^2/(2T)\}\,m(m-r)}
        {\pi\{t(T-t)\}^{3/2}}
        \exp\left\{-\frac{m^2}{2t}-\frac{(m-r)^2}{2(T-t)}\right\}\ ,
\end{equation}
for \(m\geq \max(0,r)\) and \(0<t<T\). The marginal distribution of \(\overline B\) and the conditional distribution of \(\tau^+\mid \overline B\) lead to efficient exact samplers \citep{Devroye2010,SomnathHerbei2026MaxLocation}. For a bridge from \(x\) to \(z\), one applies \eqref{eq:bridge_max_density} under \(\bbW^{0\to z-x}\) and then shifts the maximum by \(x\). We use 
\[
\BridgeMax(x,z,T)
\] 
to denote this exact sampler for $(\overline B, \tau^+)$ under \(\bbW^{x\to z}\). The minimum and its location are obtained by applying the same sampler to the reflected bridge.

For $L>0$ and $r>0$, let
$\mathbb W_{\mathrm{me},L}^{0\to r}$ denote the restricted Brownian
meander law on $[0,L]$, started at zero, conditioned to remain
nonnegative, and conditioned to end at $r$. The joint density of $(\overline B, \tau^+)$  under $\bbW^{0\to r}_{{\rm me}, T}$ derived in \citet{SomnathHerbei2026MaxLocation} has the series representation
\begin{equation}
\label{eq:restricted_meander_joint}
        p_r^{\rm me}(m,t)\propto
        \frac{1}{m^5}\sum_{j,k=1}^{\infty}(-1)^{j+k}j k^2
        \sin\left(\frac{j\pi r}{m}\right)
        \exp\left\{-\frac{k^2\pi^2 t}{2m^2}-\frac{j^2\pi^2(T-t)}{2m^2}\right\},
\end{equation}
for \(m>r\) and \(0<t<T\). Conditionally on \(\overline B=m\), the density of \(\tau^+\) is proportional to
\begin{equation}
\label{eq:Q_factorization}
        Q(t;m,r)=g(t;m)f(T-t;m,r),\qquad 0<t<T,
\end{equation}
where
\begin{align*}
        g(s;m)&=\frac{\pi^2}{m^2}\sum_{k=1}^{\infty}(-1)^{k+1}k^2
        \exp\left\{-\frac{k^2\pi^2 s}{2m^2}\right\}, \\
        f(s;m,r)&=\frac{\pi}{mr}\sum_{j=1}^{\infty}(-1)^{j+1}j
        \sin\left(\frac{j\pi r}{m}\right)
        \exp\left\{-\frac{j^2\pi^2s}{2m^2}\right\}.
\end{align*}
The function \(g(\cdot;m)\) is the exit-time density at level \(m\) for a three-dimensional Bessel process started at zero; \(f(\cdot;m,r)\) is the corresponding exit-time density for a three-dimensional Bessel process started at \(r\in(0,m)\). The factorization \eqref{eq:Q_factorization} is both analytic and probabilistic: conditional on the maximum and its location, the premaximum piece and the time-reversed postmaximum piece are independent Bessel first-passage bridges. We write
\[
        \RBMMax(r,T)
\]
for the exact sampler for \((\overline B,\tau^+)\) under $\bbW^{0\to r}_{{\rm me},T}$. This primitive is the main Brownian-extrema building block imported from \citet{SomnathHerbei2026MaxLocation}.

\section{An Exact Sampler Based on Range Reconstruction}
\label{sec:mainsampler}
We are now in the position to detail the main contribution of this paper: an EA3-type exact sampler, which avoids the layered Brownian bridge construction. Instead, it reconstructs the exact range of the Brownian proposal. The sampler uses the maximum-first version described above, because it matches directly with the restricted Brownian meander extrema algorithms. The minimum-first version is mathematically equivalent and can be implemented by reflection.

\subsection{EA3 in the style of EA2}

There is a second, more constructive way to view EA3. Start exactly as EA2 by decomposing the Brownian bridge proposal at one extremum. Suppose $Z\sim h$ and given $Z$, $B\sim \bbW^{x \to Z}$. We decompose the path $B$ at the minimum \((\underline B,\tau^-)\). The two pieces above \(\underline B\) are Bessel bridges. EA2 stops at this point because \(\underline B\) alone is enough to bound \(\phi\) under \eqref{eq:A5}. For a two-sided unbounded \(\phi\), one continues one step further: simulate the maxima of the two Bessel bridges. If those maxima are \(H_1\) and \(H_2\), then
\begin{equation*}
        \overline B=\underline B+\max(H_1,H_2),
\end{equation*}
so the exact range \([\underline B,\overline B]\) is now known. Conversely, one may decompose first at the maximum \((\overline B,\tau^+)\). The reflected pieces \(\overline B-B\) are restricted Brownian meanders, and simulating their maxima gives the depths below \(\overline B\); the minimum is then recovered exactly.

This observation is the bridge between the original layered EA3 and the algorithm proposed here. The original EA3 samples a layer containing the range. The present construction samples the exact range itself by supplementing the EA2 decomposition with exact simulation of the maxima of both Bessel-bridge pieces. Once the range and the associated sidewise extrema and their locations have been simulated, the remaining finite-dimensional sampling problem is an interval-constrained Brownian interpolation problem of the form \eqref{eq:icb_density}.

\subsection{Range reconstruction by decomposition at the maximum}

At this point we remind the reader of the general strategy in the rejection sampler: a draw $B\sim \bbZ^{x}$ must be ``revealed'' sufficiently to allow the computation of the following sequence of steps: given $B=\omega$, 
\begin{align*}
&\mbox{ calculate }\ M_\phi(\omega) 
\quad\Rightarrow\quad 
\mbox{ simulate }\ K \sim {\rm Poisson}(TM_\phi(\omega))\quad \Rightarrow\\ 
&\mbox{ simulate iid}\ S_i\sim{\rm Uniform}(0,T), i=1, \ldots, K
\quad \Rightarrow\quad 
\mbox{ reveal }\ \omega_{S_i}\ .
\end{align*}

We differ from the original EA3 construction in two respects:
(1) the Poisson bound is based on the exact realized range of $\omega$,
rather than on an enclosing layer; and (2), conditional on the simulated
extrema and their locations, the values $\omega_{S_i}$ are generated from
four interval-constrained Brownian bridge laws rather than from a layered
Brownian bridge law.

Begin with the first task. We base our strategy on a classical result of \cite{Williams1974}: the decomposition of a Brownian path. The draw $B\sim \bbZ^x$ starts with proposing the endpoint \(Z\sim h\) and, conditional on \(Z\), let \(B\sim \bbW^{x\to Z}\). 

\paragraph{Step 1:} Simulate the maximum and its location (for $B$, under $\bbW^{x\to Z}$):
$$
(\overline B, \tau^+) \leftarrow \BridgeMax(x,Z,T)
$$
and define the left and right depth processes
\begin{align}
\label{eq:left_depth}
        R_L(s)&=\overline B-B_{\tau^+-s},\qquad 0\leq s\leq \tau^+,\\
\label{eq:right_depth}
        R_R(s)&=\overline B-B_{\tau^++s},\qquad 0\leq s\leq T-\tau^+.
\end{align}

By the decomposition at the maximum for one-dimensional diffusion bridges \citep{Williams1974,PitmanYor1996}, for the joint law of
$(Z,\overline B,\tau^+)$-almost every $(z,m,t)$,
\[
\mathcal L\!\left(
R_L,R_R
\,\middle|\,
Z=z,\overline B=m,\tau^+=t
\right)
=
\mathbb W_{\mathrm{me},t}^{0\to m-x}
\otimes
\mathbb W_{\mathrm{me},T-t}^{0\to m-z}.
\]
Thus, conditional on $(Z,\overline B,\tau^+)$, the two depth processes
are independent restricted Brownian meanders with respective durations
$\tau^+$ and $T-\tau^+$ and terminal values
$\overline B-x$ and $\overline B-Z$.
Equivalently,
\begin{align*}
    \calL(R_L\mid \overline B, \tau^+, Z) &= \bbW^{0\to \overline B-x}_{{\rm me}, \tau^+},
       \\
    \calL(R_R\mid \overline B, \tau^+, Z)&=\bbW^{0 \to \overline B-Z}_{{\rm me}, T-\tau^+},
\end{align*}
conditionally independently. 
\paragraph{Step 2:} Simulate the maximum and its location for the two (left and right) pieces:
\begin{equation*}
        (H_L,\sigma_L)\leftarrow \RBMMax(\overline B-x,\tau^+),
        \qquad
        (H_R,\sigma_R)\leftarrow \RBMMax(\overline B-Z,T-\tau^+),
\end{equation*}
independently conditional on \((\overline B,\tau^+,Z)\). Here \(H_L\) and \(H_R\) are the maxima of the two depth processes. Therefore the local minima of the original Brownian bridge on the two sides of \(\tau^+\) are
\begin{equation*}
        m_L=\overline B-H_L,
        \quad t_L=\tau^+-\sigma_L,
        \qquad
        m_R=\overline B-H_R,
        \quad t_R=\tau^++\sigma_R.
\end{equation*}
The global minimum of $B$ is then
\begin{equation*}
        \underline B=\min\{m_L,m_R\},
\end{equation*}
and the exact range of the proposed path $B$ is \([\underline B,\overline B]\). Consequently, 
\begin{equation}
\label{eq:range_bound_phi}
        \sup_{\underline B\leq u\leq \overline B}\phi(u) = \sup_{0\le t\le T}\phi(\omega_t)
\end{equation}
is finite and is one valid Poisson height. As we will see, any quantity $M_\phi(\omega)$ that can be calculated and is at least as large as \eqref{eq:range_bound_phi} will suffice. This addresses the first step in the sequence above. The path decomposition described here is illustrated in Figure \ref{fig:decomp}. The left panel shows a Brownian bridge path from $x$ to $Z$, which is decomposed at its maximum $\overline{B}$ and the minimum of its left and right sub-paths. The right panels illustrate the left and right Bessel bridge depth processes $R_L$ and $R_R$, after the appropriate time reversal and time and location shift as in \eqref{eq:left_depth}, \eqref{eq:right_depth}.

\begin{figure}
    \centering
    \includegraphics[width=0.99\linewidth]{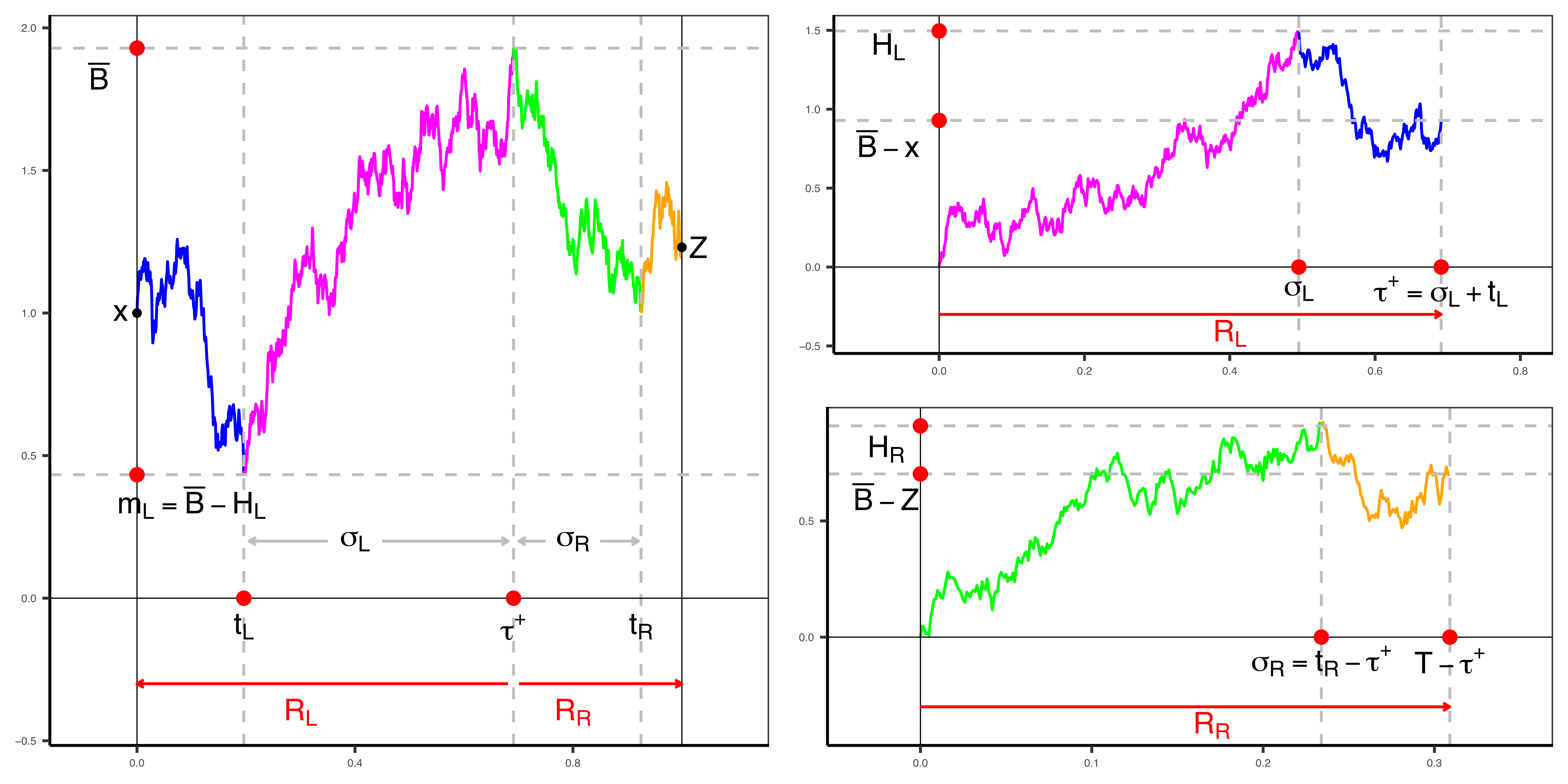}
\caption{Left panel: a Brownian bridge path decomposed at its
maximum $\overline B$ and at the minima of its left and right
subpaths. Right panels: the corresponding restricted Brownian
meander, or equivalently Bessel-bridge, depth processes $R_L$
(top) and $R_R$ (bottom), after the time reversal and shifts in
\eqref{eq:left_depth}--\eqref{eq:right_depth}.}
    \label{fig:decomp}
\end{figure}

\newpage
\subsection{Conditional law of the remaining path components}
\label{sec:cond:law}
Once a Poisson height $M_\phi(\omega)$ has been calculated and the Poisson count $K$ has been generated, it remains to explain how we reveal the values $\omega_{S_i}$ at the uniform times $S_i$, $i=1\ldots, K$. For this, we appeal to the following conditional decomposition of the proposed path $B$.

Let
\[
\mathcal G
=
\sigma\!\left\{
Z,\overline B,\tau^+,
H_L,\sigma_L,H_R,\sigma_R
\right\}.
\]
Conditional on $\mathcal G$, the following five time--value pairs are now known:
\begin{equation*}
(0,x),\qquad
(t_L,m_L),\qquad
(\tau^+,\overline B),\qquad
(t_R,m_R),\qquad
(T,Z).
\end{equation*}
The remaining path is distributed as a product of four conditionally
independent interval-constrained Brownian bridges. More explicitly,
\[
\begin{array}{c|c|c|c}
\text{piece}
&
\text{time interval}
&
\text{endpoints}
&
\text{constraint}
\\ \hline
1
&
[0,t_L]
&
x\longrightarrow m_L
&
[m_L,\overline B]
\\
2
&
[t_L,\tau^+]
&
m_L\longrightarrow\overline B
&
[m_L,\overline B]
\\
3
&
[\tau^+,t_R]
&
\overline B\longrightarrow m_R
&
[m_R,\overline B]
\\
4
&
[t_R,T]
&
m_R\longrightarrow Z
&
[m_R,\overline B].
\end{array}
\]
These are illustrated in the left panel of Figure \ref{fig:decomp}: piece 1 is the left-most piece (blue), piece 2 follows to the right (magenta), piece 3 is next (green) and piece 4 is the right-most piece (orange).
We now justify this decomposition. Conditional on
$(\overline B,\tau^+,Z)$, the two depth processes $R_L$ and $R_R$
defined in \eqref{eq:left_depth}--\eqref{eq:right_depth}  are independent restricted Brownian meanders \citep{Williams1974}.
The pair $(H_L,\sigma_L)$ is determined entirely by $R_L$, while
$(H_R,\sigma_R)$ is determined entirely by $R_R$. Consequently,
conditioning further on these two pairs preserves the conditional
independence of the left and right depth processes.

Consider first $R_L$. Conditional on $\calG$, the portion of
$R_L$ before $\sigma_L$ starts at zero, ends at $H_L$, and remains in
$[0,H_L]$. The portion after $\sigma_L$ starts at $H_L$, ends at
$\overline B-x$, and also remains in $[0,H_L]$. Using again the path-decomposition result for a Bessel bridge, we
have that, conditional on the maximum and its location, the
premaximum fragment and the time reversal of the postmaximum
fragment are independent first-passage bridge pieces. Equivalently,
in their original time directions, the two fragments have the
corresponding boundary-limit interval-constrained Brownian bridge
laws implemented by the $\operatorname{ICB}$ primitive.

The same argument applies to $R_R$. Conditional on
$\calG$, the portion before $\sigma_R$ starts at zero, ends at
$H_R$, and remains in $[0,H_R]$. The portion after $\sigma_R$ starts at
$H_R$, ends at $\overline B-Z$, and remains in $[0,H_R]$. These two
portions are conditionally independent and are interpreted as the
corresponding interval-constrained Brownian bridge pieces.

It remains to translate these four depth-process pieces back to the
original Brownian bridge. For the left side, the definition
\[
R_L(s)=\overline B-B_{\tau^+-s}
\]
reverses the direction of time. The portion of $R_L$ on
$[\sigma_L,\tau^+]$, after this time reversal, becomes the portion of
$B$ on $[0,t_L]$. Since
\[
R_L(\tau^+)=\overline B-x,
\qquad
R_L(\sigma_L)=H_L=\overline B-m_L,
\]
this gives a bridge from $x$ to $m_L$ that remains in
$[m_L,\overline B]$, which is piece~1.

Similarly, the portion of $R_L$ on $[0,\sigma_L]$, after time reversal,
becomes the portion of $B$ on $[t_L,\tau^+]$. It runs from $m_L$ to
$\overline B$ and remains in $[m_L,\overline B]$, which is piece~2.

For the right side, the definition
\[
R_R(s)=\overline B-B_{\tau^++s}
\]
does not reverse time. The portion of $R_R$ on $[0,\sigma_R]$ therefore
becomes the portion of $B$ on $[\tau^+,t_R]$. It runs from
$\overline B$ to $m_R$ and remains in $[m_R,\overline B]$, which is
piece~3. The portion of $R_R$ on $[\sigma_R,T-\tau^+]$ becomes the
portion of $B$ on $[t_R,T]$. It runs from $m_R$ to $Z$ and remains in
$[m_R,\overline B]$, which is piece~4.

All four pieces have at least one endpoint on the boundary of their
conditioning interval, and pieces~2 and~3 have both endpoints on the
boundary. These cases are interpreted through the boundary limits
underlying \eqref{eq:h_boundary}. In particular, the limiting constrained bridge laws
remain in the interior of the conditioning interval between their
endpoint times, so the extrema and their locations specified above are
preserved.

It follows that, conditional on $\mathcal G$, the four path pieces are
conditionally independent and have exactly the interval-constrained
Brownian bridge laws displayed in the table. Therefore, every value
required at a Poisson time $S_i$ can be sampled exactly using the
\ICB~primitive. The details are given in the next section.

\subsection{The range-reconstructed exact algorithm}
The entire range-reconstructed exact algorithm is presented below, from start to finish. For clarity and exactness, we explicitly state the following two assumptions:
\begin{enumerate}[label=(A\arabic*),ref=A\arabic*,leftmargin=2.2em,start=6]
\item \label{assump:A6} For any interval $[a,b]$ a valid finite bound
$$
\widehat M_\phi(a,b)\ge \sup_{a\le u\le b}\phi(u)
$$
can be computed exactly.
\item \label{assump:A7} For the unconditional algorithm, a valid exact sampler from the density $h$ is available.
\end{enumerate}
The complete procedure is now given in Algorithm \ref{alg:range_ea}.
\begin{algorithm}[!ht]
\caption{Range-reconstructed Exact Algorithm}
\label{alg:range_ea}
\begin{algorithmic}[1]
\State \textbf{Input:} initial value \(x\), horizon \(T\), drift \(\alpha\), potential \(\phi\), and lower shift \(k_L\). For the diffusion bridge version, input $z$ as well.
\State Draw \(Z\sim h\) in \eqref{eq:h_endpoint}. For the diffusion bridge version, set \(Z=z\).
\State Draw \((\overline B,\tau^+)\leftarrow\BridgeMax(x,Z,T)\).
\State Draw \((H_L,\sigma_L)\leftarrow\RBMMax(\overline B-x,\tau^+)\). Set \(m_L=\overline B-H_L\) and \(t_L=\tau^+-\sigma_L\).
\State Draw \((H_R,\sigma_R)\leftarrow\RBMMax(\overline B-Z,T-\tau^+)\). Set \(m_R=\overline B-H_R\) and \(t_R=\tau^++\sigma_R\).
\State Set \(\underline B=\min\{m_L,m_R\}\), and compute 
\[
M_\phi  = \widehat M_\phi(\underline{B}, \overline{B}) \ge \sup_{\underline{B}\ \leq\ u\ \leq\  \overline B}\phi(u)\ .
\]
\State Simulate a unit-rate Poisson process \(\Phi\) on \([0,T]\times[0,M_\phi]\). Write 
\[
\Phi=\{(S_i,V_i)\ :\ 1\leq i\leq K\}\ .
\]
\State Partition the times \(S_1,\ldots,S_K\) into the four intervals \([0,t_L]\), \([t_L,\tau^+]\), \([\tau^+,t_R]\), and \([t_R,T]\). 
Almost surely, no Poisson time equals $0, t_L, \tau^+, t_R,T$. Relabel the Poisson points, preserving each pair $(S_i, V_i)$ so that 
$
0<S_1 < \ldots < S_K<T.
$
\State Use \(\ICB\) to sample the values \(B_{S_i}\) in each interval, using the endpoints and constraints in the table above.
\State Compute \(\Sigma=\sum_{i=1}^K \ind\{V_i<\phi(B_{S_i})\}\).
\If{\(\Sigma=0\)}
    \State Accept and output the skeleton $(\calS, \Lambda)$ defined as:
    \begin{align*}
\calS &= \left(K, \{(S_i, B_{S_i}), 1\le i\le K\}\right)\\
\Lambda &= \Big\{(0,x), (t_L, m_L), (\tau^+, \overline B), (t_R, m_R), (T,Z)\Big\}
\end{align*}
\Else
    \State Reject and return to Step 2.
\EndIf
\end{algorithmic}
\end{algorithm}

Several details are worth emphasizing. First, Step 6 uses the exact range of the proposed path, which is in contrast to the original EA3 algorithm.  Assumption~\ref{assump:A6} ensures that a finite bound $M_\phi$ can be computed exactly. Any valid
upper bound will suffice; the actual supremum is not required.
Finding $M_\phi$ is problem-dependent, and its value affects the
efficiency of the sampling procedure.

Second, Step 9 is where the interval-constrained interpolation algorithms are used: the conditional laws in the four pieces are of the form \eqref{eq:icb_density}, including the boundary cases. If at Step 7 we obtain $K=0$, skip the interpolation and accept automatically. If $K>0$, without loss of generality we can assume that 
\(
0<S_1<S_2<\ldots<S_K<T\ .
\)
Let
\begin{align*}
    \calT_1 &= \{S_i\ :\ 0< S_i < t_L\}, \\
    \calT_2 &= \{S_i-t_L\ :\ t_L< S_i <\tau^+\},\\
    \calT_3 &= \{S_i-\tau^+\ :\ \tau^+< S_i < t_R\}, \\
    \calT_4 &= \{S_i-t_R\ :\ t_R< S_i < T\}\ .
\end{align*}
Then, Step 9 consists of the following four calls:
\begin{align*}
    &\ICB(t_L, x, m_L, m_L, \overline{B}; \calT_1)\\
    &\ICB(\tau^+-t_L, m_L, \overline{B}, m_L, \overline{B}; \calT_2)\\
    &\ICB(t_R - \tau^+, \overline{B}, m_R, m_R, \overline{B}; \calT_3)\\
    &\ICB(T-t_R, m_R, Z, m_R, \overline{B}; \calT_4)
\end{align*}
If any set $\calT_j, 1\le j\le 4$ is empty, the corresponding $\ICB$ call is removed.

Third, upon acceptance, the algorithm returns the skeleton
\begin{align*}
\calS &= \left(K, \{(S_i, B_{S_i}), 1\le i\le K\}\right)\\
\Lambda &= \Big\{(0,x), (t_L, m_L), (\tau^+, \overline B), (t_R, m_R), (T,Z)\Big\}
\end{align*}
Given any collection of additional time points $\calT_{\rm req} = \{ t_1, \ldots, t_M\}\subset[0,T]$, they can be inserted into the relevant intervals (which now include the points $S_i$) and subsequent calls to \(\ICB\) again, conditional on the accepted skeleton, will refine the accepted path to any degree necessary.

\begin{theorem}
Assume \textnormal{(A0)}--\textnormal{(A3)} and \textnormal{(A6)}. For the
unconditional version of the algorithm, also assume \textnormal{(A7)}.
Suppose that the samplers
\(
\BridgeMax,
\RBMMax,
\ICB
\)
return exact draws from the distributions specified in
Sections~2.3.1--2.3.2. Let $M_\phi$ denote the finite bound used in
Step~6. Then Algorithm~\ref{alg:range_ea} accepts after finitely many proposals almost surely and
returns an exact skeleton from the diffusion law $\bbQ^x$ in the sense of
\eqref{eq:skeleton}. For the diffusion bridge version, the algorithm returns an exact
skeleton from $\bbQ^{x\to z}$.

Moreover, conditional on the accepted output $(\mathcal S,\Lambda)$, the
unrevealed portions of the path can be sampled exactly at any additional
finite collection of times using further calls to the
$\operatorname{ICB}$ primitive.
\end{theorem}

\begin{proof}
We first consider the unconditional version. A convenient way to analyze
the algorithm is to imagine that an entire Brownian proposal path is drawn
first, even though the algorithm itself reveals only finitely many values
of that path. Step~2 draws $Z$ from the density $h$, and, conditionally on
$Z$, the proposal path is a Brownian bridge from $x$ to $Z$. By the
definition in \eqref{eq:Z_def}, the resulting marginal law of the
complete proposal path is $\bbZ^x$.

Step~3 draws the maximum $\overline B$ and its location $\tau^+$ from
their exact joint distribution under the Brownian bridge proposal.
Conditional on $(\overline B,\tau^+,Z)$, the two depth processes in
\eqref{eq:left_depth}--\eqref{eq:right_depth} are independent restricted
Brownian meanders. Steps~4 and~5 draw the maxima and their locations for
these two meanders from their exact conditional distributions. Therefore,
the quantities
$m_L$, $t_L$,$m_R$,$t_R$ generated by the algorithm have the same joint distribution as the
corresponding sidewise minima and their locations obtained from a complete
Brownian proposal path. In particular,
\(
\underline B=\min\{m_L,m_R\}
\)
is the exact global minimum of that proposal, while $\overline B$ is its
exact global maximum.

Conditional on the extrema and their locations, Section~\ref{sec:cond:law} identifies the
remaining path as the product of the four interval-constrained Brownian
bridge laws displayed there. Hence, after the Poisson times have been
generated, the calls to $\operatorname{ICB}$ in Step~9 sample
\(
B_{S_1},\ldots,B_{S_K}
\)
from their exact joint conditional distribution under the Brownian
proposal. Consequently, all the random quantities generated in
Steps~2--9 have the same joint distribution as they would have under the
following conceptual construction: first draw a complete path
$B\sim\bbZ^x$, compute its extrema and their locations, generate the
Poisson process, and then reveal the path only at the Poisson times.

Now fix a complete proposed path $B$. Because
\(
B_t\in[\underline B,\overline B],  0\leq t\leq T,
\)
the bound used in Step~6 satisfies
\(
\phi(B_t)\leq M_\phi, 0\leq t\leq T.
\)
Conditional on $B$, Step~7 generates a unit-rate Poisson process on
\(
[0,T]\times[0,M_\phi].
\)
The number of Poisson points lying below the graph
$t\mapsto\phi(B_t)$ is therefore Poisson distributed with mean
\(
\int_0^T\phi(B_t)\,dt.
\)
With $\Sigma$ defined as in Algorithm~\ref{alg:range_ea}, it follows that
\[
\Pr(\Sigma=0\mid B)
=
\exp\left\{
-\int_0^T\phi(B_t)\,dt
\right\}
=
a(B).
\]
This conclusion remains valid when $M_\phi=0$: in that case
$\phi(B_t)=0$ for every $t$, the Poisson process contains no points almost
surely, and the proposal is accepted automatically.

Let $D\in\mathcal F$ be any measurable collection of paths. The
distribution of the accepted proposal satisfies
\begin{align*}
\Pr(B\in D\mid\Sigma=0)
&=
\frac{\Pr(B\in D, \Sigma=0)}{\Pr(\Sigma=0)}
=
\frac{
\mathbb E_{\bbZ^x}
\left[
\mathbf 1_{\{B\in D\}}
\Pr(\Sigma=0\mid B)
\right]
}{
\mathbb E_{\bbZ^x}
\left[
\Pr(\Sigma=0\mid B)
\right]
}
\\
&=
\frac{
\mathbb E_{\bbZ^x}
\left[
\mathbf 1_{\{B\in D\}}a(B)
\right]
}{
\mathbb E_{\bbZ^x}[a(B)]
}
=
\bbQ^x(D),
\end{align*}
where the first line uses the law of iterated expectations and the final equality follows from \eqref{eq:rnqz}.
Thus the complete path associated with an accepted proposal has law
$\bbQ^x$. The values recorded in $\mathcal S$ are values of this accepted
path, and the quantities recorded in $\Lambda$ are its extrema, their
locations, and its endpoints. Hence $(\mathcal S,\Lambda)$ is an exact
finite-dimensional skeleton from $\bbQ^x$. The same argument applies to the diffusion bridge version, replacing
the random endpoint $Z$ by the fixed endpoint $z$ and replacing the
proposal law $\mathbb Z^x$ by $\mathbb W^{x\to z}$.

Finally, the probability of accepting any individual proposal is strictly
positive. Indeed,
\[
a(B)
=
\exp\left\{
-\int_0^T\phi(B_t)\,dt
\right\}
>0
\]
for every proposal path for which the integral is finite, and the integral
is finite because the proposal path has compact range and the bound in
Step~6 is finite. It follows that
\(
\mathbb E_{\bbZ^x}[a(B)]>0
\)
in the unconditional case, with the analogous statement under
$\bbW^{x\to z}$ in the bridge case. Independent repetition of the proposal
construction therefore produces an accepted proposal after finitely many
attempts almost surely.

It remains to verify the refinement statement. Write $B_{\mathrm{rem}}$
for the portions of the proposed path that are not revealed in
$(\mathcal S,\Lambda)$. Consider the joint proposal construction before
conditioning on acceptance. Once $(\mathcal S,\Lambda)$ is known, the
following quantities are known:

\[
K,\qquad
S_1,\ldots,S_K,\qquad
B_{S_1},\ldots,B_{S_K},
\]
together with the endpoints, the two sidewise minima, the maximum, and
their locations. The range $[\underline B,\overline B]$ is therefore
known, and the value of $M_\phi$ used by the algorithm is determined.

The only additional variables entering the acceptance decision are the
vertical Poisson coordinates $V_1,\ldots,V_K$. Conditional on
$(\mathcal S,\Lambda)$, these variables are independent of
$B_{\mathrm{rem}}$, and
\[
\{\Sigma=0\}
=
\bigcap_{i=1}^K
\{V_i>\phi(B_{S_i})\}
\]
depends on the path only through the already revealed values
$B_{S_1},\ldots,B_{S_K}$. Consequently, for every bounded measurable functional $F$ of
$B_{\mathrm{rem}}$,
\begin{align*}
&
\mathbb E\left[
F(B_{\mathrm{rem}})
\mathbf 1_{\{\Sigma=0\}}
\mid
\mathcal S,\Lambda
\right]
\\
&\qquad=
\mathbb E\left[
F(B_{\mathrm{rem}})
\mid
\mathcal S,\Lambda
\right]
\Pr(\Sigma=0\mid\mathcal S,\Lambda).
\end{align*}
Thus, under the joint proposal construction,
\[
\mathcal L\left(
B_{\mathrm{rem}}
\mid
\mathcal S,\Lambda,\Sigma=0
\right)
=
\mathcal L\left(
B_{\mathrm{rem}}
\mid
\mathcal S,\Lambda
\right).
\]
The left-hand side is the conditional law of the unrevealed portions of
the accepted diffusion path. The right-hand side is the corresponding
conditional law under the Brownian proposal, which is available through
interval-constrained Brownian interpolation. Acceptance therefore changes
the marginal law of the complete path from the proposal law to the target
diffusion law, but it does not introduce any further dependence on the
unrevealed path once the accepted skeleton and the auxiliary extrema have
been retained.

To refine the path at a finite collection of requested times, combine the
five time--value pairs in $\Lambda$ with the Poisson-time pairs in
$\mathcal S$, and order the resulting pairs by time within each of the four
pieces in Section~3.3. A requested time that is not already present lies
between two adjacent revealed times, say $a<t<b$, belonging to one of the
four original pieces. The values $B_a$ and $B_b$ are known.

If $[a,b]$ lies in piece~1 or piece~2, the conditional path on that
subinterval is a Brownian bridge from $B_a$ to $B_b$, conditioned to
remain in
\(
[m_L,\overline B].
\)
If $[a,b]$ lies in piece~3 or piece~4, the corresponding constraint is
\(
[m_R,\overline B].
\)
Hence, if
\(
a<t_1<\cdots<t_j<b
\)
are the requested times lying in this particular subinterval, their values
are generated by the call
\[
\operatorname{ICB}
\left(
b-a,\,
B_a,\,
B_b,\,
\ell,\,
u;\,
\{t_1-a,\ldots,t_j-a\}
\right),
\]
where $[\ell,u]$ is the constraint belonging to the original piece.
Separate calls are made for requested times lying between different pairs
of adjacent revealed points.

These calls sample from the exact conditional law of the accepted path.
No additional Poisson rejection test is required during refinement. The
newly generated time--value pairs may then be added to $\mathcal S$, and
the same argument may be applied repeatedly to refine the path at any
further finite collection of times.
\end{proof}

\section{Simulations}
\label{sec:sims}

Algorithm~\ref{alg:range_ea} depends on several exact sampling primitives. The first is \(Z\sim h\), which is model-specific and is also required by EA1--EA3. In many examples \(h\) is available in closed form or can be sampled by a simple one-dimensional rejection sampler. The second is \(\BridgeMax\), based on \eqref{eq:bridge_max_density},
and the third is \(\RBMMax\), based on \eqref{eq:restricted_meander_joint}--\eqref{eq:Q_factorization}.
Both extrema primitives ultimately reduce to one-dimensional rejection samplers with series-method comparisons, as in \citet{SomnathHerbei2026MaxLocation}. Fourth, the $\ICB$ routines
\citep{HerbeiSomnath2026Interval} 
are used to construct the proposed skeleton, followed by the Poisson accept/reject test.

The computation of \(M_\phi\) in Step 6 is deterministic conditional on \([\underline{B},\overline{B}]\). If \(\phi\) is a polynomial or a simple analytic function, the supremum can often be computed exactly by checking endpoints and stationary points. Overestimation of $M_\phi$ is allowed. It maintains exactness of Algorithm~\ref{alg:range_ea} and does not change the acceptance probability, but it does increase the expected Poisson count and hence the computational cost.

All the simulations below are conducted on a standard desktop Mac Studio computer with $32\ {\rm GB}$ of memory.

\subsection{Ornstein--Uhlenbeck diffusion}
\label{sec:ou-example}

Consider the Ornstein--Uhlenbeck diffusion
\begin{equation}
    dX_t=(\theta_1-\theta_2X_t)\,dt+dW_t,
    \qquad X_0=x,
    \qquad \theta_1\in\mathbb R,\quad \theta_2>0.
    \label{eq:ou-sde}
\end{equation}
This model is analytically solvable and therefore provides a useful validation
example for Algorithm~\ref{alg:range_ea}. The purpose of this example is to provide a setting in which every
model-specific component of Algorithm~\ref{alg:range_ea} is explicit and the distribution of the accepted output can be checked against known finite-dimensional distributions.
Let
\[
    \mu_\infty=\frac{\theta_1}{\theta_2},
    \qquad
    m(t)=\mu_\infty+(x-\mu_\infty)e^{-\theta_2t},
    \qquad
    q(t)=\frac{1-e^{-2\theta_2t}}{2\theta_2}.
\]
The exact transition law for the OU diffusion \eqref{eq:ou-sde} is
\begin{equation}
    X_t\mid X_0=x
    \sim
    N\{m(t),q(t)\}.
    \label{eq:ou-transition}
\end{equation}
Following the notation from Section \ref{sec:prelim},
\[
    \alpha(u)=\theta_1-\theta_2u,
    \qquad
    A(u)=\theta_1u-\frac{\theta_2u^2}{2}.
\]
We have
\begin{equation}
    \frac{1}{2}\{\alpha^2(u)+\alpha'(u)\}
    =
    \frac{1}{2}(\theta_1-\theta_2u)^2-\frac{\theta_2}{2}.
    \label{eq:ou-lower-bound}
\end{equation}
The exact global lower bound in \eqref{eq:ou-lower-bound} is attained at
\(u=\theta_1/\theta_2\). We therefore use
\(
    k_L=-\theta_2/2.
\)
The resulting potential is
\begin{equation}
    \phi(u)
    =
    \frac{1}{2}(\theta_1-\theta_2u)^2
    =
    \frac{\theta_2^2}{2}
    \left(u-\frac{\theta_1}{\theta_2}\right)^2.
    \label{eq:ou-potential}
\end{equation}
Thus \(\phi(u)\to\infty\) as \(u\to\pm \infty\). Consequently,  neither EA1 nor either one-sided version of EA2 is applicable.

\paragraph{Sampling the proposal endpoint.}

The endpoint density in \eqref{eq:h_endpoint} is
\[
    h(z)
    \propto
    \exp\left\{
        \theta_1z-\frac{\theta_2z^2}{2}
        -\frac{(z-x)^2}{2T}
    \right\}.
\]

Therefore, 
\begin{equation*}
    Z\sim
    N\left(
        \frac{x+\theta_1T}{1+\theta_2T},
        \frac{T}{1+\theta_2T}
    \right).
\end{equation*}

\paragraph{The range-dependent Poisson bound.}

Suppose that Steps 3--5 of Algorithm~\ref{alg:range_ea} have reconstructed the exact proposal
range
\(
    [\underline B,\overline B].
\)
Because the potential in \eqref{eq:ou-potential} is convex, its maximum on
any compact interval is attained at one of the two endpoints. Hence the
exact bound required in Step 6 is
\begin{equation*}
    M_\phi
    =
    \widehat M_\phi(\underline B,\overline B)
    =
    \frac{1}{2}
    \max\left\{
        (\theta_1-\theta_2\underline B)^2,
        (\theta_1-\theta_2\overline B)^2
    \right\}.
\end{equation*}

\paragraph{Exact acceptance probability.}

The acceptance probability also has a closed form in this example. The
normalizing constant in the endpoint density is
\[
    c(x,T)
    =
    \sqrt{\frac{2\pi T}{1+\theta_2T}}
    \exp\left\{
        -\frac{x^2}{2T}
        +
        \frac{(x+\theta_1T)^2}
             {2T(1+\theta_2T)}
    \right\}.
\]
From \eqref{eq:rnqz} the (outer) acceptance rate of Algorithm~\ref{alg:range_ea} is 
\begin{align*}
     p_{\rm acc}(x,T) &= \mathbb E_{\mathbb Z^x}
    \left[
        \exp\left\{
            -\int_0^T\phi(B_s)\,ds
        \right\}
    \right]
    =
    \frac{\sqrt{2\pi T}}{c(x,T)}
    \exp\{A(x)+k_LT\}, \\
   &= \sqrt{1+\theta_2T}
    \exp\left\{
        -\frac{\theta_2T}{2}
        -
        \frac{T(\theta_1-\theta_2x)^2}
             {2(1+\theta_2T)}
    \right\}.
\end{align*}
This provides a check of the complete
implementation, including endpoint generation, range reconstruction,
interval-constrained interpolation, and the Poisson test.

\paragraph{Numerical configuration.}
We use two scenarios: (1) $\theta_1=2, \theta_2=1$ (standard) and (2) $\theta_1 = 6, \theta_2 = 3$ (strong reversion), with the initial value set to $x=1.5$.

Beginning with the standard setup,  and setting the time increment $\Delta=1$, we generate the sample recursively: for each $i=1,\ldots,n$ and $m=0,\ldots,5$, Algorithm~\ref{alg:range_ea} is used to simulate
$X_{(m+1)\Delta}^{(i)}$ using $X_{m\Delta}^{(i)}$ as the initial value. Thus the reported samples at $T=1,\ldots,6$ are obtained by chaining one-step Markov transitions. 
The top panels of Figure~\ref{fig:OU} display histograms and kernel density estimates for $X_T$ based on $n=10,000$ draws from Algorithm~\ref{alg:range_ea} at $T\in \{1.0,  6.0\}$, as well as the theoretical transition density \eqref{eq:ou-transition}.  Table~\ref{tab:ou:1} reports numerical summaries for all time points $T = 1,\ldots, 6$, but only for the last transition $X_{T-\Delta}^{(i)} \to X_T^{(i)}$, $i=1, \ldots, n$. Columns $4-6$ report the average number of proposals per accepted (terminal)
increment,
\(
\overline P = (1/n)\sum_{i=1}^n P_i,
\)
the mean Poisson count per accepted (terminal) increment,
\(
\overline K_{\rm tot} = \frac{1}{n}\sum_{i=1}^n\sum_{j=1}^{P_i} K_{i,j},
\)
 where $K_{i,j}$ is the Poisson count for proposal $j$ in the generation of accepted increment $i$,  and the standard score $D_T$ is defined as
\[
D_T =
\frac{\overline X_T - m(T)}
{\sqrt{q(T)/n}},
\qquad\mbox{ where }\ \ 
\overline X_T = \frac{1}{n}\sum_{i=1}^n X^{(i)}_T\ .
\]
Column $7$ is the empirical-to-exact variance ratio
$$
\frac{S_T^2}{q(T)}, 
\qquad\mbox{ where }\ \ 
S_T^2 = \frac{1}{n-1}\sum_{i=1}^n(X_T^{(i)} - \overline X_T)^2\ .
$$
The next two columns report the effective predicted and empirical
acceptance probabilities at the last transition $X_{T-\Delta}^{(i)} \to X_T^{(i)}$,
\[
\widehat p_{\rm acc}^{\;\rm pred}(\Delta,T)
=
\left[
\frac{1}{n}\sum_{i=1}^n
\frac{1}{p_{\rm acc}(X_{T-\Delta}^{(i)},\Delta)}
\right]^{-1},
\qquad
\widehat p_{\rm acc}^{\;\rm emp}(\Delta,T)
=
\frac{n}{\sum_{i=1}^n P_i}
=
\overline P^{-1}.
\]
These numerical summaries are consistent with the exact transition law and
provide numerical validation for the implementation of Algorithm~\ref{alg:range_ea}.
We repeat the experiment with a strong reversion setup: $\theta_1 = 6, \theta_2 = 3$. In this case, the OU diffusion \eqref{eq:ou-sde} has the same stationary mean, but a smaller stationary variance, with a strong ``pull'' towards the mean driven by $\theta_2 = 3$.  Figure~\ref{fig:OU}, lower panels and Table~\ref{tab:ou:1}, strong-reversion scenario, display the same graphical and numerical summaries as in the standard setup. We again note the strong agreement between theoretical and empirical quantities.

\begin{figure}[!ht]
    \centering
    \includegraphics[width=0.9\linewidth]{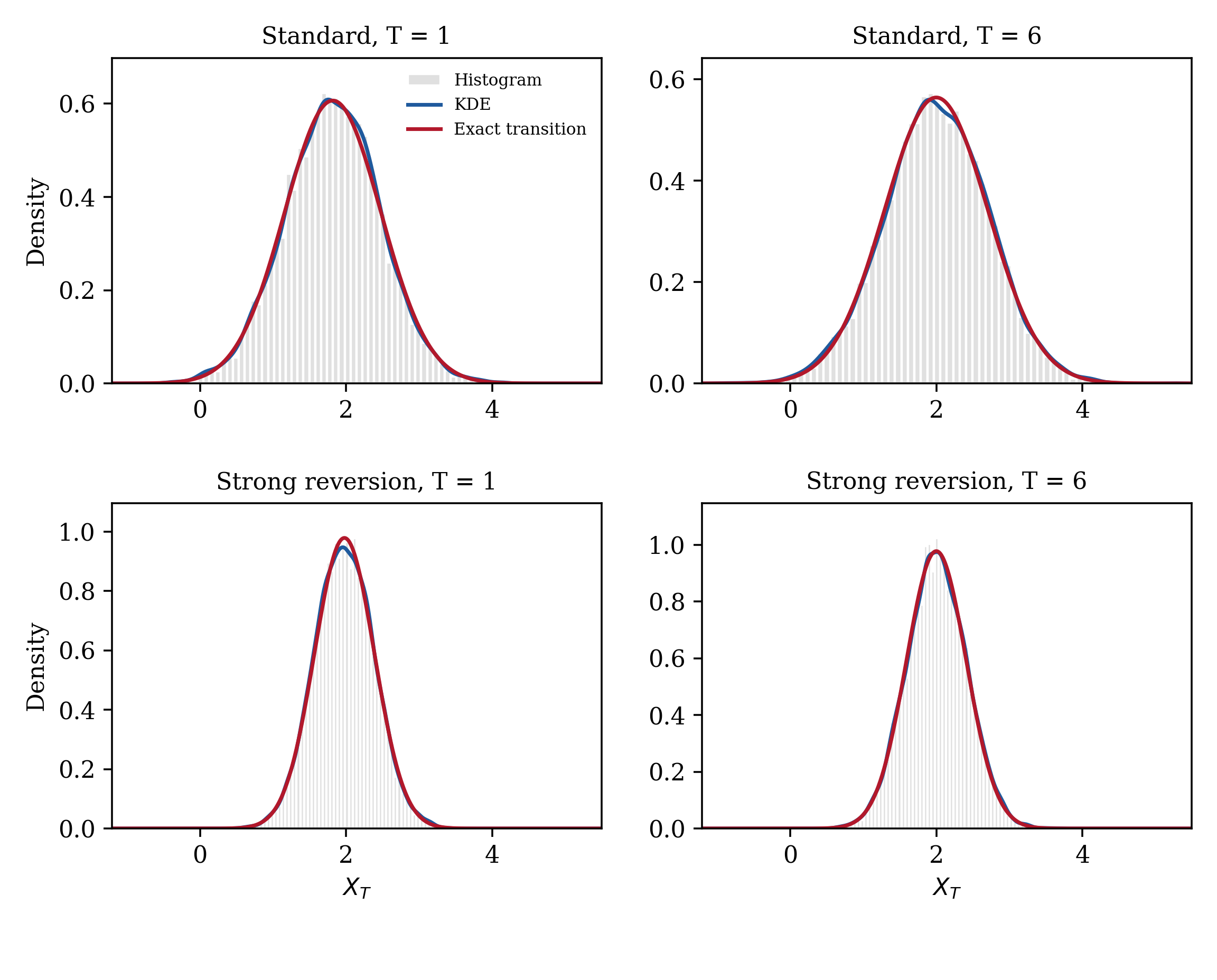}
    \caption{Distributions generated by Algorithm~\ref{alg:range_ea} for the Ornstein--Uhlenbeck diffusion, standard case (top panels) and strong-reversion case (bottom panels). Each distribution is obtained by chaining exact Markov increments of length $\Delta=1$. The left panels display the law of $X_1$ and the right panels display the law of $X_6$. Gray histograms show simulated $X_T$ values, for $T = 1,6$, blue curves show kernel density estimates, and red curves show the exact OU transition densities.}
    \label{fig:OU}
\end{figure}

\begin{table}[!ht]
\centering
\small
\caption{Diagnostics for stepwise Algorithm~\ref{alg:range_ea} simulation of the Ornstein--Uhlenbeck diffusion. Each row summarizes the terminal transition $X_{T-\Delta}\to X_T$ and uses $n=10{,}000$ accepted Markov increments ending at time $T$. Here $\Delta$ is the increment length, $\overline{P}$ is the mean number of proposals per accepted terminal increment, $\overline{K}_{\rm tot}$ is the mean total Poisson count (accepted and rejected proposals) per accepted terminal increment, $D_T$ is the standard score for $X_T$ based on the $n$ samples, $S_T^2/q(T)$ is the empirical-to-exact variance ratio, and $\widehat{p}_{\rm acc}^{\,\mathrm{pred}}(\Delta,T)$ and $\widehat{p}_{\rm acc}^{\,\mathrm{emp}}(\Delta,T)$ are the predicted and empirical acceptance probabilities for the last increment.}
\label{tab:ou:1}
\vspace{5pt}
\begin{tabular}{@{}ccccccccc@{}}
\toprule
Scenario & $T$ & $\Delta$ & $\overline{P}$ & $\overline{K}_{\rm tot}$ & $D_T$ & $S_T^2/q(T)$ & $\widehat p_{\rm acc}^{\mathrm{pred}}(\Delta,T)$ & $\widehat{p}_{\rm acc}^{\mathrm{emp}}(\Delta, T)$  \\
\midrule
standard & 1.0 & 1.0 & 1.234 & 0.961 & 0.01 & 0.978 & 0.8058 & 0.8104 \\
standard & 2.0 & 1.0 & 1.336 & 1.349 & -0.58 & 1.000 & 0.7532 & 0.7485 \\
standard & 3.0 & 1.0 & 1.346 & 1.397 & 1.16 & 0.983 & 0.7443 & 0.7430 \\
standard & 4.0 & 1.0 & 1.340 & 1.373 & -0.19 & 0.989 & 0.7455 & 0.7465 \\
standard & 5.0 & 1.0 & 1.347 & 1.397 & 0.07 & 1.004 & 0.7442 & 0.7422 \\
standard & 6.0 & 1.0 & 1.333 & 1.363 & 0.29 & 1.016 & 0.7432 & 0.7503 \\
\midrule
strong reversion & 1.0 & 1.0 & 2.967 & 17.308 & -0.46 & 0.990 & 0.3369 & 0.3370 \\
strong reversion & 2.0 & 1.0 & 2.819 & 16.589 & -0.35 & 1.008 & 0.3543 & 0.3547 \\
strong reversion & 3.0 & 1.0 & 2.906 & 17.485 & -0.01 & 0.999 & 0.3528 & 0.3442 \\
strong reversion & 4.0 & 1.0 & 2.830 & 16.721 & 1.51 & 0.993 & 0.3542 & 0.3534 \\
strong reversion & 5.0 & 1.0 & 2.818 & 16.681 & -0.59 & 1.006 & 0.3533 & 0.3549 \\
strong reversion & 6.0 & 1.0 & 2.828 & 16.841 & 1.13 & 1.003 & 0.3520 & 0.3536 \\
\bottomrule
\end{tabular}
\end{table}

To assess the computational cost for the various components of Algorithm~\ref{alg:range_ea}, we ran a separate stationary-start timing experiment. For each scenario,
we sampled the initial value of each one-step increment from the stationary OU
distribution,
\(
X_0 \sim N\left(\theta_1/\theta_2,1/(2\theta_2)\right),
\)
and then generated one accepted increment of length $\Delta=1$ using
Algorithm~\ref{alg:range_ea}. The experiment was repeated over $R=100$ independent
replicates with $n=500$ accepted increments per replicate. Our results are displayed in Table~\ref{tab:ou:2}.

The results show that the strong-reversion case is more expensive for two
related reasons. First, the empirical acceptance probability decreases from
approximately $0.74$ in the standard case to amake pproximately $0.35$ in the
strong-reversion case, so more proposals, and hence more calls to \RBMMax,
are required per accepted increment. Second, the larger potential in the
strong-reversion case leads to substantially larger Poisson workloads: the
mean total number of Poisson points generated per accepted increment
increases from about $1.4$ to about $17$. The share of runtime spent on
interval-constrained interpolation, $\phi$-evaluation, and acceptance testing
rises from $2.3\%$ to $10.1\%$. Range reconstruction through \RBMMax~remains
the dominant cost, accounting for $94.8\%$ of the mean runtime in the
standard case and $80.5\%$ in the strong-reversion case.

Across the $R=100$ replicates, the mean wall-clock time per accepted increment was
$5.1$ ms for the standard case and $12.7$ ms for the strong-reversion case.
The largest component timing variance is associated with \RBMMax~in the
standard case and with \BridgeMax~in the strong-reversion case, followed
in the latter by \RBMMax. These large variances are consistent with
occasional low-acceptance configurations in the internal rejection
samplers, where prolonged calls can contribute disproportionately to
the replicate-level average runtimes.

\begin{table}[!ht]
\centering
\small
\caption{Timing buckets for the stationary-start one-step OU benchmark, using $n=500$ accepted increments of length $\Delta=1$ per replicate. The sample mean and sample variance are computed across $R=100$ replicate-average times per accepted increment. Percentages are $100$ times the ratio of the unrounded mean bucket time to the unrounded mean total time for the same scenario.}
\label{tab:ou:2}
\vspace{5pt}
\begin{tabular}{@{}llrrr@{}}
\toprule
Scenario & Bucket & Mean ms/draw & Variance $(\mathrm{ms}/\mathrm{draw})^2$ & \% of mean \\
\midrule
standard & endpoint proposal & 0.005 & $1.01\times 10^{-7}$ & 0.1 \\
standard & \BridgeMax & 0.131 & 0.065 & 2.6 \\
standard & \RBMMax & 4.845 & 8.417 & 94.8 \\
standard & $\phi$ bound + Poisson & 0.008 & $2.20\times 10^{-7}$ & 0.1 \\
standard & \ICB + $\phi$ + accept & 0.119 & 0.002 & 2.3 \\
standard & other & 0.001 & $3.27\times 10^{-9}$ & $<0.1$ \\
\midrule
Total &  & 5.109 &  &  \\
\midrule
strong reversion & endpoint proposal & 0.010 & $3.21\times 10^{-7}$ & 0.1 \\
strong reversion & \BridgeMax & 1.154 & 49.864 & 9.1 \\
strong reversion & \RBMMax & 10.180 & 17.027 & 80.5 \\
strong reversion & $\phi$ bound + Poisson & 0.022 & $2.42\times 10^{-6}$ & 0.2 \\
strong reversion & \ICB + $\phi$ + accept & 1.283 & 0.516 & 10.1 \\
strong reversion & other & 0.002 & $8.66\times 10^{-9}$ & $<0.1$ \\
\midrule
Total &  & 12.652 &  &  \\
\bottomrule
\end{tabular}
\end{table}


\subsection{A nonlinear double-well diffusion}
\label{sec:double-well}

We next consider the double-well Langevin diffusion
\begin{equation}
    dX_t
    =
    \gamma\left(X_t-X_t^3\right)\,dt+dW_t,
    \qquad
    X_0=x,
    \qquad
    \gamma>0.
    \label{eq:double-well-sde}
\end{equation}
The deterministic drift has stable equilibrium points at $-1$ and $1$ and
an unstable equilibrium point at $0$. The parameter $\gamma$ controls the
strength of the attraction toward the two wells. This example is useful
because its potential is unbounded in both directions, so neither EA1 nor
either one-sided version of EA2 is applicable. The drift function is
\begin{equation}
    \alpha(u)=\gamma(u-u^3)=\gamma u(1-u^2),
    \label{eq:double-well-drift}
\quad
\Rightarrow
\quad 
A(u)
    =
    \int_0^u \alpha(v)\,dv
    =
    \frac{\gamma}{4}
    -
    \frac{\gamma}{4}(u^2-1)^2 \le \frac{\gamma}{4}
\end{equation}
It is immediate that the normalizing constant defining the endpoint density $h(\cdot)$ is
finite for every $x\in\mathbb R$, $T>0$, and $\gamma>0$. The drift function is smooth and locally Lipschitz. We verify that \eqref{eq:double-well-sde} has a  nonexplosive solution and satisfies Assumption \eqref{assump:rn} in Appendix~\ref{app:a0}.

The stationary density for \eqref{eq:double-well-sde} is proportional to
\begin{equation}
    \pi_\gamma(u)
    \propto
    \exp\{2A(u)\}
    \propto
    \exp\left\{
        -\frac{\gamma}{2}(u^2-1)^2
    \right\},
    \qquad u\in\mathbb R,
    \label{eq:double-well-stationary}
\end{equation}
which makes the two-well structure explicit. An elementary closed-form transition density for \eqref{eq:double-well-sde} is not available. We have 
\begin{align}
    \Psi_\gamma(u)
    &:=
    \frac12\{\alpha^2(u)+\alpha'(u)\} = 
    \frac{\gamma^2}{2}u^2(1-u^2)^2
    +
    \frac{\gamma}{2}(1-3u^2).
    \label{eq:double-well-psi}
\end{align}
Standard calculus gives that 
$$
\Psi_\gamma(u) \ge \Psi_\gamma(\sqrt{y_+}),
\qquad
\mbox{ where we define }\quad
y_\pm = \frac{2 \pm\sqrt{1+9/\gamma}}{3}\ ,
$$
thus, we can take
$$
k_L = \Psi_\gamma(\sqrt{y_+}) = \frac{\gamma}{2}
    \left[
        \gamma y_+(1-y_+)^2+1-3y_+
    \right].
$$
and, as before, we have
\begin{equation}
\label{eq:double-well-phi-definition}
\phi(u) = \Psi_\gamma(u) - k_L =  \frac{\gamma^2}{2}
    (u^2-y_+)^2
    \bigl(u^2+2y_+-2\bigr)\ge 0\ ,
\end{equation}
since $y_+>1$. Furthermore,
\begin{equation*}
    \phi(u)
    \sim
    \frac{\gamma^2}{2}u^6,
    \qquad |u|\longrightarrow\infty.
\end{equation*}
The potential is therefore unbounded in both tails. This is a genuine
EA3 setting.

\paragraph{Sampling the proposal endpoint.}

The endpoint density under the biased Brownian proposal is
\begin{align}
    h(z)
    &\propto
    \exp\left\{
        A(z)-\frac{(z-x)^2}{2T}
    \right\}
    \propto
    \exp\left\{
        -\frac{(z-x)^2}{2T}
        -
        \frac{\gamma}{4}(z^2-1)^2
    \right\},
    \qquad z\in\mathbb R,
    \label{eq:double-well-endpoint}
\end{align}
where a multiplicative factor $\exp(\gamma/4)$ has been omitted. Equation \eqref{eq:double-well-endpoint} gives a simple exact rejection
sampler for the endpoint density $h(\cdot)$:
\begin{enumerate}
    \item Sample $\sZ\sim N(x,T)$; $\sU\sim {\rm Uniform}(0,1)$;
    \item Accept $\sZ$ if
    \begin{equation*}
    \sU \le \exp\left\{
            -\frac{\gamma}{4}
            \bigl(\sZ^2-1\bigr)^2
        \right\}
    \end{equation*}
\end{enumerate}
The acceptance probability of the endpoint sampler is
\begin{equation}
    p_h(x,T,\gamma)
    =
    \mathbb E\left[
        \exp\left\{
            -\frac{\gamma}{4}
            \bigl(\sZ^2-1\bigr)^2
        \right\}
    \right],
    \qquad \sZ\sim N(x,T).
    \label{eq:double-well-endpoint-acceptance}
\end{equation}
Although this expectation does not generally have an elementary closed
form, it is one-dimensional and can be evaluated accurately when a
theoretical benchmark for the endpoint-sampler acceptance rate is desired.

\paragraph{The range-dependent Poisson bound.}
Let $[a,b]$ be any compact interval. Differentiating
\eqref{eq:double-well-phi-definition} gives
\begin{equation*}
    \phi'(u)
    =
    \gamma u
    \left\{
        \gamma(1-u^2)(1-3u^2)-3
    \right\}.
\end{equation*}
The stationary points are therefore
\(
    0,\pm\sqrt{y_+}
\)
together with
\(
    \pm\sqrt{y_-}
\)
when $\gamma>3$. Define
\begin{equation*}
    \mathcal S_\gamma
    =
    \{0,-\sqrt{y_+},\sqrt{y_+}\}
    \cup
    \begin{cases}
        \{-\sqrt{y_-},\sqrt{y_-}\},
            & \gamma>3,\\[1mm]
        \emptyset,
            & 0<\gamma\leq3,
    \end{cases}
\end{equation*}
and
\begin{equation*}
    \mathcal C_\gamma(a,b)
    =
    \{a,b\}
    \cup
    \bigl(\mathcal S_\gamma\cap[a,b]\bigr).
\end{equation*}
Since $\phi$ is a polynomial, its exact supremum on $[a,b]$ is
\begin{equation*}
    \widehat M_\phi(a,b)
    =
    \max_{u\in\mathcal C_\gamma(a,b)}
    \phi(u).
\end{equation*}
Thus the range bound needed by Algorithm~\ref{alg:range_ea} requires only finitely many
evaluations of the explicit polynomial in
\eqref{eq:double-well-phi-definition}. No numerical optimization is required. In particular, suppose that Algorithm~\ref{alg:range_ea} reconstructs the exact Brownian
proposal range
\(
    [\underline B,\overline B].
\)
The exact Poisson height is then
\begin{equation*}
    M_\phi
    =
    \widehat M_\phi(\underline B,\overline B)
    =
    \max_{u\in
        \mathcal C_\gamma(\underline B,\overline B)}
    \phi(u).
\end{equation*}

\paragraph{Exact acceptance probability.}
The normalizing constant for the endpoint distribution is
\[
    c(x,T)
    =
    \int_{\mathbb R}
    \exp\left\{
        A(z)-\frac{(z-x)^2}{2T}
    \right\}\,dz\ .
\]
As before, the general
acceptance identity gives that the acceptance probability is
\begin{equation*}
    p_{\mathrm{acc}}(x,T)
    =
    \frac{\sqrt{2\pi T}}{c(x,T)}
    \exp\{A(x)+k_LT\}.
\end{equation*}
By \eqref{eq:double-well-drift} and
\eqref{eq:double-well-endpoint-acceptance},
\[
    c(x,T)
    =
    \sqrt{2\pi T}\,
    e^{\gamma/4}
    p_h(x,T,\gamma).
\]
Consequently,
\begin{equation*}
    p_{\mathrm{acc}}(x,T)
    =
    \frac{
        \exp\{A(x)+k_LT-\gamma/4\}
    }{
        p_h(x,T,\gamma)
    }.
\end{equation*}
This formula provides a one-dimensional numerical benchmark for the
empirical diffusion-level acceptance probability. 

\paragraph{Numerical configuration.}

For the numerical illustration, we use \(\gamma = 1, x = -1\).
Figure~\ref{fig:double-well:1} compares Algorithm~\ref{alg:range_ea} with the
Euler--Maruyama discretization. At each horizon \(T\in\{0.5,1,2,4\}\), the
Algorithm~\ref{alg:range_ea} distribution is represented by a histogram and
kernel density estimate based on \(10{,}000\) exact endpoint draws. The
sample at \(T=0.5\) was generated using a single exact transition of length
\(0.5\), whereas the samples at \(T=1,2,4\) were obtained recursively from
exact Markov transitions of length \(\Delta=0.2\), with each accepted
endpoint used as the initial value for the subsequent transition. For
comparison, \(10{,}000\) Euler--Maruyama paths were simulated over \([0,4]\)
for each discretization \(\delta=2^{-i}\), \(i=1,\ldots,5\), retaining the
values at all four horizons. Within each panel, all kernel density
estimates use the bandwidth selected from the corresponding
Algorithm~\ref{alg:range_ea} sample, so the curves use a common degree of
smoothing. The coarsest Euler schemes produce pronounced distortion and
excessive dispersion, with numerical instability becoming especially
severe at the longer horizons. At $T=4$, the curves for $\delta=1/2$ and
$\delta=1/4$ are conditional kernel density estimates based on the finite
Euler--Maruyama endpoints within the exact sample range, normalized by
the number retained. These comprise $82.79\%$ and $99.26\%$ of the paths,
respectively. For the finest discretization ($\delta=2^{-5}$), the density estimates are
nearly indistinguishable from those of the exact simulation across all
four horizons.

\begin{figure}[!ht]
    \centering
    \includegraphics[width=0.9\linewidth]{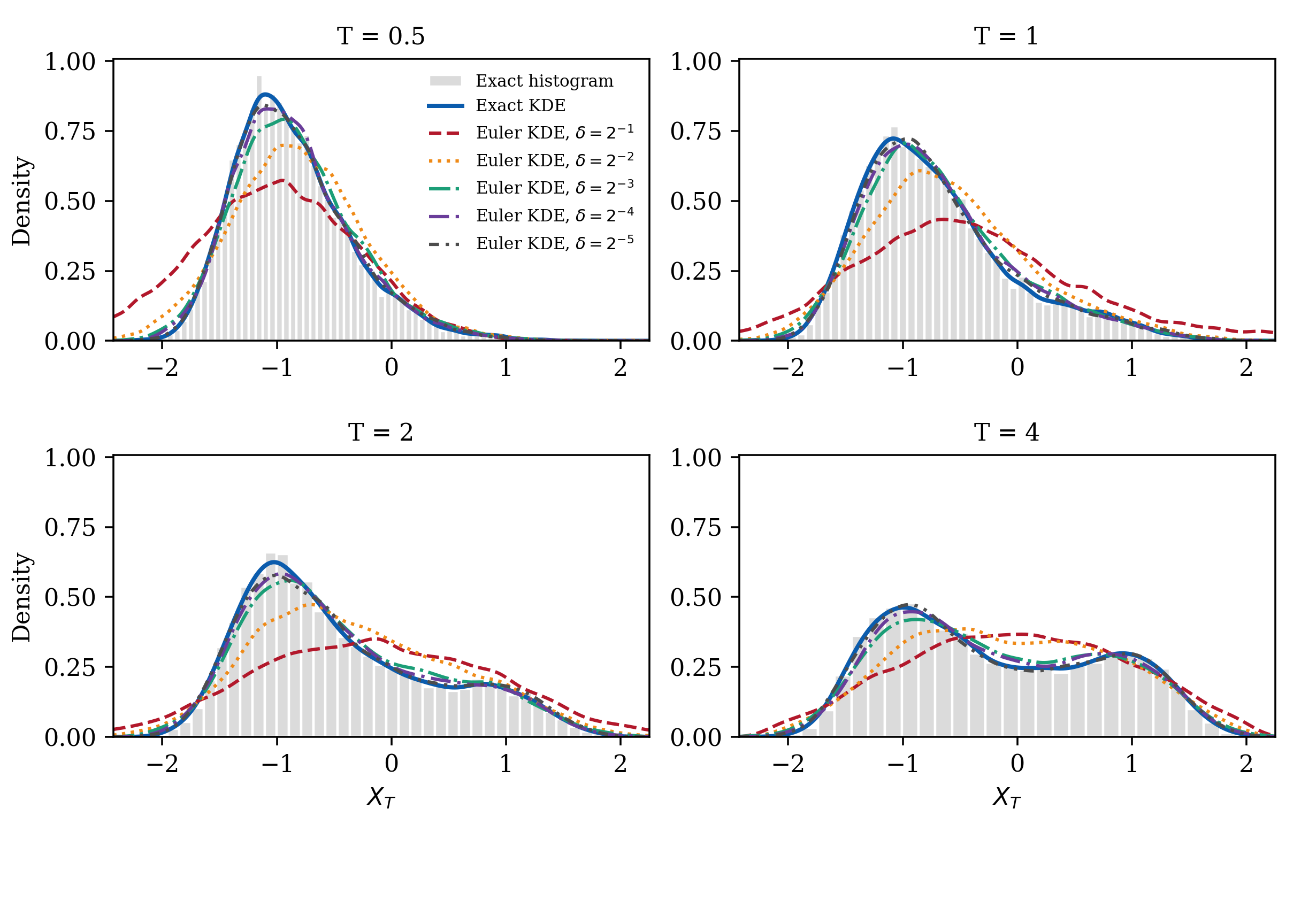}
    \caption{Terminal distributions for the double-well diffusion with $\gamma=1$ and $X_0=-1$ at $T\in\{0.5,1,2,4\}$. Each gray histogram and solid blue kernel density estimate uses $n=10{,}000$ Algorithm~\ref{alg:range_ea} endpoints. For $T>0.5$, exact paths are constructed from Markov increments of length $\Delta=0.2$. The colored curves are Euler--Maruyama kernel density estimates using $10{,}000$ paths for each $\delta=2^{-i}$, $i=1,\ldots,5$. At $T=4$, the curves for $\delta=1/2$ and $\delta=1/4$ are conditional density estimates of finite endpoints within the exact sample range, normalized by the retained counts: $8{,}279$ ($82.79\%$) and $9{,}926$ ($99.26\%$), respectively.}
\label{fig:double-well:1}
\end{figure}

Table~\ref{tab:double-well:1} displays numerical summaries for the terminal transition $X_{T-\Delta}^{(i)} \to X_T^{(i)}$, $i=1,\ldots,n$. Similarly to the OU model, we display the average number of proposals per accepted (terminal) increment $\overline P$, the mean Poisson count per accepted (terminal) increment $\overline K_{\rm tot}$, and the sample variance ratio $S^2_{\rm exact}/S^2_{{\rm EM}, \delta}$. When $T = 4$ and $\delta = 1/2, 1/4$ the Euler discretization scheme becomes unstable and produces \texttt{Inf} terminal values, and thus the variance ratio is omitted. The last two columns report the predicted and empirical acceptance probabilities for the last transition, defined as
$$
\widehat p_{\rm acc}^{\;\rm pred}(\Delta, T) = 
\left[
\frac{1}{n}\sum_{i=1}^n\frac{1}{p_{\rm acc}(X_{T-\Delta}^{(i)}, \Delta)}
\right]^{-1},
\quad
\widehat p_{\rm acc}^{\;\rm emp}(\Delta, T) = \overline P^{-1}\ .
$$
When evaluating $p_{\rm acc}(x,\Delta)$, we compute the Gaussian expectation
$p_h$ in \eqref{eq:double-well-endpoint-acceptance} using adaptive Simpson
quadrature in Python, with an explicit bound on the omitted Gaussian tails.

\begin{table}[!ht]
\centering
\scriptsize
\caption{Diagnostics for stepwise simulation of Algorithm~\ref{alg:range_ea} for the double-well diffusion with $\gamma=1$ and $X_0=-1$. For $T>0.5$, paths are generated by exact Markov increments of length $\Delta$. Each row summarizes the terminal transition $X_{T-\Delta}\to X_T$ and uses $n=10{,}000$ terminal draws.  Here $\overline{P}$ is the mean number of proposals per accepted terminal increment, $\overline{K}_{\rm tot}$ is the mean total Poisson count (accepted and rejected proposals) per accepted terminal increment. The ratio $S^2_{\rm exact}/S^2_{\rm EM,\delta}$ is the Algorithm~\ref{alg:range_ea}-to-Euler variance ratio, and $\widehat{p}^{\rm\, pred}_{\rm acc}(\Delta,T)$ and $\widehat{p}^{\,\rm emp}_{\rm acc}(\Delta,T)$ are the predicted and empirical effective acceptance probabilities for the terminal increment. A variance ratio is omitted when an Euler sample contains non-finite endpoints.}
\label{tab:double-well:1}
\vspace{5pt}
\resizebox{\textwidth}{!}{%
\begin{tabular}{@{}ccccccccccc@{}}
\toprule
$T$ & $\Delta$ & $\overline P$ & $\overline K_{\rm tot}$ & \multicolumn{5}{c}{$S_{\rm exact}^2/S_{{\rm EM},\delta}^2$} & $\widehat p_{\rm acc}^{\; \rm pred}(\Delta, T)$ & $\widehat p_{\rm acc}^{\; \rm emp}(\Delta, T)$ \\
\cmidrule(lr){5-9}
 & & & & $\delta=2^{-1}$ & $\delta=2^{-2}$ & $\delta=2^{-3}$ & $\delta=2^{-4}$ & $\delta=2^{-5}$ & & \\
\midrule
0.5 & 0.5 & 1.695 & 3.346 & 0.518 & 0.731 & 0.895 & 0.973 & 0.982 & 0.5931 & 0.5898 \\
1.0 & 0.2 & 1.280 & 0.773 & 0.214 & 0.192 & 0.959 & 0.988 & 0.967 & 0.7748 & 0.7814 \\
2.0 & 0.2 & 1.282 & 0.730 & $1.60\times 10^{-19}$ & $1.24\times 10^{-291}$ & 0.967 & 0.991 & 0.960 & 0.7726 & 0.7799 \\
4.0 & 0.2 & 1.299 & 0.760 & -- & -- & 1.003 & 1.018 & 0.981 & 0.7705 & 0.7699 \\
\bottomrule
\end{tabular}%
}
\end{table}

To assess the computational cost of Algorithm~\ref{alg:range_ea} for the
double-well diffusion, we conducted a stationary-start timing experiment
analogous to the OU timing study. With \(\gamma=1\),
the initial value of each increment was sampled independently from the
invariant density \eqref{eq:double-well-stationary} using a rejection
sampler with a standard normal proposal. We performed $R=100$ independent
replicates, each comprising \(n=500\) accepted Algorithm~\ref{alg:range_ea}
increments of length \(\Delta=0.2\). For every accepted increment, cumulative
wall-clock time over all accepted and rejected proposals was partitioned
into six computational components: endpoint proposal, \BridgeMax,
\RBMMax, evaluation of the \(\phi\)-bound and generation of the Poisson
process, interval-constrained interpolation together with evaluation of
\(\phi\) and the acceptance test, and residual operations. The mean time
per accepted increment was computed separately within each replicate;
Table~\ref{tab:double-well-algorithm1-stationary-timing} reports the mean
and sample variance of these replicate-level averages, together with each
component's mean time as a percentage of the mean total wall-clock time.

The mean total time is $5.447$ ms per accepted increment. The computational
cost is dominated by \RBMMax, which requires $5.246$ ms on average and
accounts for $96.3\%$ of the mean runtime. Interpolation, evaluation of
\(\phi\), and the acceptance test account for $1.4\%$, while
\BridgeMax~accounts for $2.0\%$; all remaining components collectively contribute less
than $1\%$. Although absolute timings depend on the implementation and
hardware, this allocation identifies restricted-meander range
reconstruction through \RBMMax~as the principal computational bottleneck
for this configuration. Its sample variance of replicate-average times,
$48.054\,(\mathrm{ms}/\mathrm{draw})^2$, is also much larger than those of
the other components. This variability is consistent with occasional
low-acceptance configurations in the internal rejection sampler, which
can require many proposals even when the diffusion-level acceptance
probability is moderate.

\begin{table}[!ht]
\centering
\small
\caption{Stationary-start timing breakdown for one-step Algorithm~\ref{alg:range_ea} simulation of the double-well diffusion with $\gamma=1.0$, $n=500$ accepted increments per replicate, and $\Delta=0.2$. Mean and sample variance per increment are computed across $R=100$ replicate averages; the final column gives each bucket's mean as a percentage of the mean total wall-clock time per accepted increment.}
\label{tab:double-well-algorithm1-stationary-timing}
\vspace{5pt}
\begin{tabular}{@{}lrrr@{}}
\toprule
Bucket & Mean ms/draw & Variance $(\mathrm{ms}/\mathrm{draw})^2$ & \% of mean \\
\midrule
endpoint proposal & 0.005 & $2.41\times 10^{-7}$ & 0.1 \\
\BridgeMax & 0.107 & 0.014 & 2.0 \\
\RBMMax & 5.246 & 48.054 & 96.3 \\
$\phi$ bound + Poisson & 0.012 & $1.34\times 10^{-6}$ & 0.2 \\
\ICB + $\phi$ + accept & 0.076 & 0.002 & 1.4 \\
other & $9.57\times 10^{-4}$ & $4.00\times 10^{-9}$ & $1.76\times 10^{-2}$ \\
\midrule
Total & 5.447 &  &  \\
\bottomrule
\end{tabular}
\end{table}

\section{Discussion}
\label{sec:disc}

This paper develops an EA3-type exact simulation method based on direct
reconstruction of the realized range of a Brownian bridge proposal. The key
observation is that the extremum decomposition used by EA2 can be continued
one step further. After decomposing at the Brownian maximum, the two
resulting restricted-meander depth processes are supplemented by their own
maxima and maximum locations. This recovers the sidewise minima and hence
the exact global range. Conditional on the resulting extrema and their
locations, all remaining randomness is carried by four independent
interval-constrained Brownian bridges. The construction therefore replaces
the layered Brownian-bridge interpolation required by the original EA3 with
a modular combination of exact extremum simulation and exact constrained
interpolation.

The realized range is the smallest closed interval containing the proposal
path. Consequently, when the Poisson height is chosen as the exact supremum
of $\phi$ over the available interval, the range-based height cannot exceed
the corresponding height obtained from any enclosing EA3 layer. This does
not, by itself, imply a reduction in total computing time. The acceptance
probability is determined by the change of measure and is unaffected by the
choice of a valid upper bound; the bound controls only the number of Poisson
points that must be generated and interpolated. Direct range reconstruction
requires two calls to the restricted-meander extremum sampler on every
proposal. The numerical experiments illustrate this tradeoff. Range
reconstruction through \RBMMax~accounts for most of the runtime in both
examples: $94.8\%$ and $80.5\%$ in the standard and strong-reversion OU
scenarios, respectively, and $96.3\%$ in the double-well example. The
larger Poisson workload in the strong-reversion OU case also raises the
share spent on interval-constrained interpolation, $\phi$-evaluation, and
acceptance testing from $2.3\%$ to $10.1\%$. The method should therefore be
viewed as a probabilistically direct alternative to the layer construction,
rather than as uniformly faster than existing exact algorithms.

The scope of the method is governed by explicit analytical and
computational conditions. After any required Lamperti transformation, the
model must satisfy the change-of-measure and lower-bound assumptions used by
the Exact Algorithm framework. The unconditional version additionally
requires an exact sampler for the endpoint density, and the potential must
admit a computable finite bound on every compact interval. These
model-specific tasks are elementary in the two examples considered here but
need not be so in general. The algorithm also relies on exact implementations
of \BridgeMax, \RBMMax, and \ICB. The correctness
result guarantees distributional exactness and almost-sure termination under
the stated assumptions; it is not a uniform finite-cost or complexity
result. For long horizons or strongly confining drifts, applying the
algorithm over shorter Markov increments may be substantially more
effective than attempting a single rejection step over the complete time
interval.

Several directions for further work follow directly from the computational
decomposition. Improvements to the restricted-meander extremum sampler
would reduce the principal cost observed in the double-well experiment.
Adaptive selection of the increment length, sequential generation of the
Poisson test, and hybrid constructions that choose between exact-range and
layer-based bounds according to the local configuration may also improve
efficiency while preserving exactness. The retained extrema and the exact
refinement property make the output suitable for diffusion-bridge
simulation, path-dependent Monte Carlo calculations, and likelihood or
data-augmentation methods based on exact diffusion skeletons. Extensions to
time-inhomogeneous scalar diffusions would require bounds for a
space--time-dependent potential, while a multidimensional analogue would
require a replacement for the ordered one-dimensional range and is
therefore a more fundamental problem. The main result here is that, in one
dimension, the full two-sided range information required for exact Poisson
thinning can be constructed directly and retained in a form that leaves the
remaining path exactly simulable.

\section{Statements and Declarations}

\paragraph{Code availability.} The code used for the numerical experiments
will be made available on GitHub at \url{https://github.com/herbei/Range_EA3}.

\bibliographystyle{plainnat}
\bibliography{references}

\appendix

\section{Verification of the assumptions for the double-well diffusion}
\label{app:a0}

Consider the double-well diffusion
\[
    dX_t=\gamma(X_t-X_t^3)\,dt+dW_t,
    \qquad X_0=x,
    \qquad \gamma>0,
\]
with drift
\[
    \alpha(u)=\gamma(u-u^3).
\]
We verify the standing assumptions used in Section~4.2.

\paragraph{Existence, uniqueness, and nonexplosion.}
The drift $\alpha$ is continuously differentiable and locally Lipschitz. The
standard localization of the Lipschitz uniqueness theorem therefore gives a
pathwise unique maximal solution up to its explosion time; see
\citep[Section 5.2]{KaratzasShreve1998}.

To verify nonexplosion, we apply the Lyapunov criterion of
\citet[Theorem~3.5]{Khasminskii2012}. The generator is
\[
    \mathcal{L}f(u)
    =
    \alpha(u)f'(u)+\frac12 f''(u).
\]
For
\[
    V(u)=1+u^2,
\]
we have $V(u)\to\infty$ as $|u|\to\infty$ and
\[
\begin{aligned}
    \mathcal{L}V(u)
    &=
    1+2u\alpha(u) 
    =
    1+2\gamma(u^2-u^4) \\
    &\le
    1+\frac{\gamma}{2}
    \le
    \left(1+\frac{\gamma}{2}\right)V(u),
\end{aligned}
\]
where $u^2-u^4\le 1/4$. The Lyapunov criterion therefore implies that
the maximal solution is nonexplosive. Pathwise uniqueness then implies
uniqueness in law; see
\citep[Section 5.3]{KaratzasShreve1998}. Hence the
double-well equation has a global, weakly unique solution.

\paragraph{Assumptions \textup{(A1)--(A3)}.}
Assumption~\eqref{assump:diff} follows from $\alpha\in C^\infty(\mathbb{R})$.
As calculated in equation~\eqref{eq:double-well-drift},
\[
    A(u)
    =
    \int_0^u\alpha(v)\,dv
    =
    \frac{\gamma}{4}
    -
    \frac{\gamma}{4}(u^2-1)^2
    \le
    \frac{\gamma}{4}.
\]
Consequently,
\[
\begin{aligned}
    c(x,T)
    &=
    \int_{\mathbb{R}}
    \exp\left\{
        A(u)-\frac{(u-x)^2}{2T}
    \right\}\,du \\
    &\le
    e^{\gamma/4}
    \int_{\mathbb{R}}
    \exp\left\{
        -\frac{(u-x)^2}{2T}
    \right\}\,du
    =
    e^{\gamma/4}\sqrt{2\pi T}
    <
    \infty,
\end{aligned}
\]
which verifies Assumption~\eqref{assump:endpoint}. Equations~\eqref{eq:double-well-psi}--\eqref{eq:double-well-phi-definition} show that
\[
    \Psi_\gamma(u)
    :=
    \frac12\{\alpha^2(u)+\alpha'(u)\}
    \ge k_L>-\infty,
\]
and therefore verify Assumption~\eqref{assump:lower}.

\paragraph{Assumption \textup{(A0)}.}
Under $\mathbb{W}^x$, define the stochastic exponential
\[
    L_t
    =
    \exp\left\{
        \int_0^t\alpha(B_s)\,dB_s
        -
        \frac12\int_0^t\alpha^2(B_s)\,ds
    \right\},
    \qquad 0\le t\le T.
\]
Since Brownian paths are bounded on compact time intervals and $\alpha$ is
continuous, $L$ is a well-defined nonnegative local martingale. Itô's formula,
as used in equation~\eqref{eq:rnqw}, gives
\[
    L_t
    =
    \exp\left\{
        A(B_t)-A(x)
        -
        \int_0^t\Psi_\gamma(B_s)\,ds
    \right\}.
\]
Using $A(u)\le\gamma/4$ and $\Psi_\gamma(u)\ge k_L$, we obtain the
deterministic bound
\[
    0<L_t
    \le
    \exp\left\{
        \frac{\gamma}{4}-A(x)+T|k_L|
    \right\},
    \qquad 0\le t\le T.
\]
Thus $L$ is a bounded, and hence uniformly integrable, martingale. By
Girsanov's theorem
\citep[Theorem~3.5.1, p.~191]{KaratzasShreve1998}, the probability
measure defined on $\mathcal{F}_T$ by
\[
    \frac{d\widetilde{\mathbb{Q}}^x}
         {d\mathbb{W}^x}
    =
    L_T
\]
makes
\[
    B_t-x-\int_0^t\alpha(B_s)\,ds,
    \qquad 0\le t\le T,
\]
a Brownian motion. Hence the coordinate process solves the double-well
equation under $\widetilde{\mathbb{Q}}^x$. By uniqueness in law,
$\widetilde{\mathbb{Q}}^x=\mathbb{Q}^x$, and therefore
\[
    \frac{d\mathbb{Q}^x}{d\mathbb{W}^x}
    =
    \exp\left\{
        \int_0^T\alpha(B_s)\,dB_s
        -
        \frac12\int_0^T\alpha^2(B_s)\,ds
    \right\}.
\]
This verifies Assumption~\eqref{assump:rn}.

The compact-localization argument described by
\citet[Section~6]{BeskosPapaspiliopoulosRoberts2008} provides an alternative
justification: on any event confining the path to a compact interval, the
drift is bounded and the stopped Girsanov formula applies. In the present
construction, that compact interval is supplied directly by the reconstructed
range $[\underline B,\overline B]$.

\end{document}